\documentclass[12pt]{amsart}
\usepackage[hmargin=0.8in,height=8.6in]{geometry}
\usepackage{amssymb,amsthm, times, enumerate}
\usepackage{delarray,verbatim}
\usepackage{bm}
\usepackage{ifpdf}
\ifpdf
\usepackage[pdftex]{graphicx}
 \else
\usepackage[dvips]{graphicx}
 \fi

\usepackage{ifpdf}
\usepackage{color}
\definecolor{webgreen}{rgb}{0,.5,0}
\definecolor{webbrown}{rgb}{.8,0,0}
\definecolor{emphcolor}{rgb}{0.95,0.95,0.95}
\usepackage{hyperref}
\hypersetup{
          colorlinks=true,
          linkcolor=webbrown,
          filecolor=webbrown,
          citecolor=webgreen,
          breaklinks=true}
\ifpdf \hypersetup{pdftex,
            pdfstartview=FitH,
            bookmarksopen=true,
            bookmarksnumbered=true
} \else \hypersetup{dvips} \fi

\newcommand {\ph}{\hat{p}}
\newcommand {\ah}{\hat{\alpha}}
\newcommand {\pcheck}{\tilde{p}}
\newcommand {\acheck}{\tilde{\alpha}}

\renewcommand{\S}{\mathcal{S}}

\numberwithin{equation}{section}

\newtheorem{theorem}{Theorem}[section]
\newtheorem{proposition}{Proposition}[section]

\newtheorem{remark}{Remark}[section]
\newtheorem{lemma}{Lemma}[section]

\numberwithin{remark}{section} \numberwithin{proposition}{section}
\numberwithin{corollary}{section}
\newcommand {\R}{\mathbb{R}}

\newcommand {\F}{\mathcal{F}}

\newcommand {\p}{\mathbb{P}}

\newcommand {\E}{\mathbb{E}}
\newcommand {\PP}{\mathbb{P}}
\def\1{1{\hskip -3.3 pt}\hbox{I}}
\def\lv{{L\'{e}vy}}

\newcommand{\diff}{{\rm d}}

\newcommand{\lev}{L\'{e}vy }

\newcommand {\lap}{\zeta}
\newcommand{\lapinv}{\Phi(r)}

\title[American Step-Up and Step-Down Default Swaps]{American Step-Up and Step-Down  Default Swaps  under L\'{e}vy Models*}
\author[T. Leung]{Tim S.T. Leung$^\dag$\,}\thanks{$\dag$ \tiny{IEOR Department, Columbia University, New York NY 10027, USA. Email: \mbox{leung@ieor.columbia.edu}}}

\author[K. Yamazaki]{\,Kazutoshi Yamazaki$^\ddag$\,}\thanks{$\ddag$\,\tiny{Center for the Study of Finance and Insurance,
Osaka University, 1-3 Machikaneyama-cho, Toyonaka City, Osaka 560-8531, Japan. Email: \mbox{k-yamazaki@sigmath.es.osaka-u.ac.jp}} }

\thanks{\tiny{*{The authors would like to thank Prof.\ Goran Peskir and the anonymous referees for reading the draft and for their helpful remarks and suggestions.} This work is partially supported by NSF grant DMS-0908295, by Grant-in-Aid for Young Scientists (B) No.\ 22710143, the Ministry of Education, Culture, Sports, Science and Technology, and  by Grant-in-Aid for Scientific Research (B) No.\  2271014, Japan Society for the Promotion of Science.}}

\begin{document}
\begin{abstract} This paper studies the valuation of a class of  default swaps with
the embedded option to switch to a different premium and notional
principal anytime prior to a credit event. These are early exercisable
contracts that give the protection buyer or seller the right to step-up,
step-down, or cancel the swap position. The pricing problem is formulated
under a structural credit risk  model based on \lev processes. This leads to
the analytic and numerical studies of several  optimal stopping problems
subject to early termination due to default. In a general spectrally negative
\lev model,  we rigorously derive the  optimal exercise strategy. This allows
for instant computation of the credit spread under various specifications.
Numerical examples are provided to examine the impacts of default risk and
contractual features on the credit spread and exercise strategy. \noindent
\end{abstract}
\tiny{\date{Accepted: September 6, 2012. First submission: December 25, 2010}.}
\maketitle \noindent \small{\textbf{Keywords:}\,  optimal stopping; credit default swaps; step-up and step-down options; \lev processes; {scale functions} }

\section{Introduction}
Credit default swaps (CDSs) are among the most liquid and widely used instruments for managing and transferring credit risks. Despite the recent market turbulence,  their market size still exceeds US\$30 trillions\footnote{\tiny{According to the ISDA Market Survey, the total CDS outstanding volume in 2009 is US\$30,428 billions.}}. In a standard single-name CDS, the protection buyer pays a pre-specified periodic premium (the CDS spread) to the protection seller  to cover the loss of the face value of an asset if the reference entity defaults before expiration. The contract   stipulates that both the buyer and seller have to commit to their respective positions until the default time or expiration date. To modify the initial CDS exposure in the future, one common way is to acquire appropriate positions later from the market, but it is subject to credit spread fluctuations and market illiquidity, especially during adverse market conditions.

To provide additional flexibility to investors, credit default swaptions and other derivatives on  CDSs have emerged.  For instance, the payer (receiver) default swaption is a European option that gives the holder the right to buy (sell) protection at a pre-specified strike spread at expiry, given that default has not occurred. Otherwise, the  swaption is knocked out. See, for example, \cite{HullWhiteCDswaption03}.  By appropriately combining a default swaption with a vanilla CDS position, one can create a callable or putable default swap. A callable (putable) CDS allows the protection buyer (seller) to terminate the contract at some fixed future date. Hence, as described here, the callable/putable CDSs are in fact \emph{cancellable} CDSs.  Typically, the callable feature is   paid for through incremental premium on top of the standard CDS spread, so selling a callable CDS can enhance the yield from the seller's perspective.

In this paper, we consider a class of default swaps embedded with an option for the investor (protection buyer or seller) to adjust the premium and notional amount once for a pre-specified fee prior to default. Specifically, these non-standard contracts  equip the standard default swaps with the early exercisable rights such as (i) the \emph{step-up option} that allows the investor to increase the protection and premium at exercise, and (ii) the \emph{step-down option} to  reduce the protection and premium. By definition, these contracts are indeed generalized versions of the callable and putable CDSs mentioned above, and thus are more flexible credit risk management tools. Henceforth, we shall use the more general meaning of the terminology callable and putable default swaps, rather than limiting them to cancellable CDSs.

The main contribution of our paper is to determine the credit spread for these default swaps under a \lev model, and analyze the optimal strategy for the buyer or seller to exercise the step-up/down option. Specifically, we model the default time as the \emph{first passage time} of a \lev process representing some underlying asset value. We decompose the default swap with step-up/down option into a combination of an American-style credit default swaption and  a vanilla default swap. From the investor's perspective, this gives rise to an optimal stopping problem subject to possible sudden early termination from default risk.  Our formulation is based on a general \lev process, and then we solve analytically for a general \emph{spectrally negative} \lev process. By employing the \emph{scale function} and other properties of \lev processes, we derive analytic characterization for the optimal exercising strategy.  This in turn allows for a highly efficient computation of the credit spread for these  contracts. We provide a series of numerical examples to illustrate the credit spread behavior and optimal exercising strategy under various contract specifications and scenarios.

We adopt a \lv-based structural credit risk model that extends the original approach introduced by Black and Cox \cite{BlackCox76}  where the asset value follows a geometric Brownian motion.  Other structural  default models based on \lev and other jump processes can  also be found in  \cite{schoutensCDS07,helberenkrogers,zhou2001}. To our best knowledge, the valuation of American step-up and step-down default swaps has not been studied elsewhere.  For \lv-based  pricing  models for other credit derivatives, such as European credit default swaptions and collateralized debt obligations (CDOs), we highlight \cite{AsmussenMadanPistorius07,EberleinKlugeSchb06,schoutenjosson}, among others.

\lev processes have been widely applied in  derivatives pricing. Some well-known examples of \lev pricing models include the variance gamma (VG) model   \cite{Madan_1998},  the normal inverse Gaussian (NIG) model \cite{Barndorff_1998}, the CGMY model \cite{CGMY}  as well as a number of  jump diffusion models (see \cite{Kou_Wang_2004,Merton_1976}). In this paper, instead of focusing on a particular type of \lev process, we consider a general class of \lev processes with only negative jumps. This is called the spectrally negative \lev process and has been drawing much attention recently, as a generalization of the classical Cram\'er-Lundberg and other compound-Poisson type processes. A number of fluctuation identities can be expressed in terms of the scale function and are used in a number of applications.  We refer the reader to \cite{Alili2005,Avram_2004} for derivatives pricing, \cite{Kyprianou_Surya_2007} for optimal capital structure, \cite{Baurdoux2008,Baurdoux2009}  for stochastic games, \cite{Avram_et_al_2007,Kyprianou_Palmowski_2007,Loeffen_2008} for optimal dividend problem, and \cite{Egami_Yamazaki_2010} for optimal timing of capital reinforcement.  For a comprehensive account, see \cite{Kyprianou_2006}.

A key part of our analysis focuses on a non-standard American option  subject to default risk (see Proposition \ref{prop-V}).   We discuss both the perpetual and finite-maturity cases.  The former is related to some existing work on  perpetual early exercisable options under various \lev models, for example \cite{Asmussen_2004,Avram_2004,Boyarchenko_QF04,levendorski_QF04,mordecki_FS02}. The infinite horizon nature  provides significant convenience for analysis and sometimes leads to  explicit solutions.  Working under a general {spectrally negative} \lev model, we provide  analytic  results for  the timing strategies and contract values.   For numerical examples, we select the phase-type (and hyperexponential) fitting approach by Egami and Yamazaki \cite{Egami_Yamazaki_2010_2} to illustrate the cases when the process is a mixture of Brownian motion and a compound Poisson process with Pareto-distributed jumps. We then apply our formulation and  results to study the finite-maturity case.  For finite-maturity American options under \lev models,  the pricing problem  typically requires  numerical solutions to the underlying partial integral differential equation (PIDE), or other simulation methods; see, among others, \cite{avramchanusabel02,hirsaMadan, jaimungalSukov08}. In our paper, we illustrate how to approximate the finite-maturity case  using our analytical solutions to the perpetual case.

The rest of the paper is organized as follows. In Section \ref{sect-overview}, we formulate the default swap valuation problems under a general \lev model. In Section \ref{sect-buyer-sol}, we focus on the spectrally negative \lev model and provide a complete solution and detailed analysis. Section \ref{section_numerical} provides the numerical results. In Section  \ref{section_finite_maturity}, we apply the results to the finite-maturity case.  Section \ref{section_conclusion} concludes the paper and presents some extensions of our model.  Most proofs are included in the Appendix.

\section{Problem Overview}\label{sect-overview}
Let $(\Omega, \F, \p)$ be a complete probability space, where  $\p$ is the \emph{risk-neutral} measure used for pricing. We assume there exists a L\'{e}vy process $X=\{X_t;\, t\ge 0\}$, and denote by $\mathbb{F}=(\F_t)_{t\ge 0}$ the filtration generated by $X$. The value of the reference entity (a company stock or other assets) is assumed to evolve according to an \emph{exponential \lev process} $S_t = e^{X_t}$, $t\geq0$. Following the Black-Cox \cite{BlackCox76} structural approach, the default event is triggered by $S$ crossing a lower level $D$, so the default time is given by the first passage time: $\theta_D :=\inf\{\,t\ge 0\,:\, X_t\,  \leq \,\log D\,\}$. Without loss of generality, we can take $\log D =0$ by shifting the initial value $x$. Henceforth, we shall work with the default time:
\[\theta:=\inf\{\,t\ge 0\,:\, X_t\,  \leq \,0\,\},\]
where we assume $\inf \emptyset = \infty$.  Throughout this paper, we denote by  $\p^x$ the probability law and $\E^x$ the expectation under which $X_0=x$.

\subsection{Credit Default Swaps and Swaptions} In preparation for default swaps with step-up/down options, let us start with the basic concepts of credit default swaps and swaptions. Under  a $T$-year CDS on a unit face value, the protection buyer pays a constant premium payment \$$p$ continuously over time until default time $\theta$ or maturity $T$, whichever comes first. If default occurs before  $T$, the buyer will receive the default payment $\alpha := 1-R$ at time $\theta$, where  $R$ is the assumed constant recovery rate (typically 40\%).   From the buyer's perspective, the expected discounted payoff is given by
 \begin{align}\bar{C}(x, T;p,\alpha) :=\E^x \left[ -\int_0^{\theta\wedge T} e^{-rt} p\, \diff t + \alpha e^{-r \theta} 1_{\{ \theta \le T\}}   \right],  \label{def_C_bar}
\end{align}
where $r > 0$ is the positive constant risk-free interest rate. The quantity $\bar{C}(x, T;p,\alpha)$ can be viewed as the market price for the buyer to enter (or long) a CDS with an agreed premium $p$, default payment $\alpha$ and maturity $T$. On the opposite side of the trade, the protection seller's expected cash flow is  $-\bar{C}(x,T;p,\alpha) =\bar{C}(x,T;-p,-\alpha) \in\R$.

In standard practice, the CDS spread $\bar{p}$ is determined at inception such that $\bar{C}(x,T;\bar{p},\alpha) =0$,
yielding zero expected cash flows for both  parties. Direct calculations show that the credit spread can be expressed as
\begin{align}
\bar{p}(x, T;\alpha)= \frac{\alpha \, r\,\lap_T(x) }{1-\lap_T(x) - e^{-rT} \p^x\{\theta >T\} },\quad \,\text{where}\quad \lap_T(x) := \E^x \left[  e^{-r\theta} 1_{\{\theta \leq T\}}\right].
\label{def_p_bar} \end{align}

For most \lev models, due to the lack of explicit formulas, the computation of the CDS spread is  based on  simulation or other approximation methods  (see, for example, \cite{schoutensCDS07}). Alternatively, one can consider the perpetual case as an approximation  and to obtain analytic or explicit bounds. This is a popular approach adopted for equity derivatives, especially American options, for which the finite-maturity contracts do not admit closed-form solutions while the perpetual versions often do (see \cite{Boyarchenko_QF04, levendorski_QF04, mordecki_FS02} for some examples under  \lev models).

To illustrate, we set $T=+\infty$ and express the buyer's CDS price as
  \begin{align}C(x;p,\alpha) &:=\E^x \left[ -\int_0^{\theta } e^{-rt} p\, \diff t + \alpha\, e^{-r \theta}     \right] \label{Cxp}  \\
  &~\,=  \left(\frac{p}{r} +\alpha\right) \lap(x) -\frac{p}{r},\notag
\end{align}
where
\begin{align}
 \lap(x) := \E^x \left[  e^{-r\theta}\right]. \label{zeta}
\end{align}
The seller's CDS price is $-C(x;p,\alpha) =C(x;-p,-\alpha) \in \R$. Solving $C(x;p,\alpha) =0$ yields the credit spread:
\begin{align}
p(x; \alpha)=  \frac{\alpha\, r\,\lap(x) }{1-\lap(x) }.
\label{def_p_pertpetual} \end{align}Therefore, the credit spread calculation reduces to computing the Laplace transform $\lap(x)$, which admits an explicit analytic formula under some well-known \lev models (see \eqref{laplace_theta} below for the spectrally negative case). It is clear from \eqref{def_p_pertpetual} that the CDS spread scales linearly in $\alpha$: $p(x; \alpha) =\alpha\,  p(x; 1)$.

Next, we introduce a  \emph{perpetual} American payer and receiver default swaptions, which give the holder the right to, respectively, buy and sell protection on a perpetual CDS with default payment $a$ at a pre-specified spread $\kappa$ for the strike price $K$ upon exercise. If default occurs prior to exercise, then the swaption is knocked out and becomes worthless.   The payer and receiver swaption holder is required to pay an upfront fee, which is given by respectively
 \begin{align}v(x;\kappa,a, K) &:=\sup_{\tau \in \S}\E^x \left[ e^{-r \tau}  \left(C(X_{\tau}; \kappa,a) -K\right)^+  1_{\{ \tau < \theta \}}  \right], \quad \text{ and }\label{Amer_v}\\
u(x;\kappa,a,K) &:=\sup_{\tau \in \S}\E^x \left[ e^{-r \tau}  \left(-C(X_{\tau}; \kappa,a)-K\right)^+  1_{\{ \tau < \theta \}}  \right], \label{Amer_u}
\end{align}where
\begin{align}
\S : = \left\{ \mathbb{F}\textrm{-stopping time } \tau: \tau \leq \theta \; \text{ a.s. }\right\} \label{def_S}
\end{align}  is the set of all $\mathbb{F}$-stopping times \emph{smaller than or equal to} the default time. The two price functions are related by
\begin{align}
v(x;\kappa,a, K) = u(x;-\kappa,-a, K). \label{equal_u_v}
\end{align}
  In summary, $v(x;\kappa,a, K)$ is the payer default swaption price when $\kappa, a \geq0$, and it is the receiver default swaption price when $\kappa, a \leq 0$.

  \begin{remark} We remark that the \emph{perpetual} American payer and receiver default swaptions introduced above are non-standard option contracts, but they bear similarity to the traditional European-style default swaptions. In Section \ref{section_finite_maturity}, we will discuss the finite-maturity version of these contracts.
  \end{remark}

\subsection{American Callable Step-Up and Step-Down Default Swaps}\label{sect-callableCDS}
Next, we consider a default swap contract with an embedded option that permits
the protection buyer to change the face value and premium once for a fee. We   discuss   the perpetual case here, and  the finite-maturity case in Section \ref{section_finite_maturity}.
Beginning from initiation, the buyer pays a premium $p$ for a protection of
a unit face value. At any time prior to default, the buyer can select a time
$\tau$ to switch to a new contract with a new premium $\ph$ and face
value $q$ for a fee $\gamma \geq 0$. The default payment then changes from
$\alpha$ to $\ah = q \alpha$ after the exercise time $\tau$. Here,
$p, \alpha, \ph, \ah$, and $\gamma$ are constant non-negative parameters
pre-specified at time zero. The  buyer's maximal expected cash flow is given
by
\begin{align}&V(x; p,\ph, \alpha, \ah, \gamma)  \label{valuefn}\\
 &:=\sup_{\tau \in \S} \E^x \left[ -\int_0^\tau e^{-rt} p\,\diff t + 1_{\{\tau < \infty\}}\bigg( -\int_\tau^\theta e^{-rt}\ph\,\diff t -e^{-r \tau} \gamma  1_{\{\tau <\theta \}} + e^{-r \theta}(\ah 1_{\{\tau <\theta\}} +\alpha 1_{\{\tau =\theta \}}) \bigg) \right], \notag
\end{align}
with  $\S$   defined   in \eqref{def_S}.

This formulation covers default swaps with the following provisions:
\begin{enumerate}\item \emph{Step-up Option}: if $\ph> p$ and $\ah >\alpha$, then the buyer is allowed to increase the coverage once from $\alpha$ to $\ah$ by paying the fee $\gamma$ and a higher premium $\ph$ thereafter.
\item \emph{Step-down Option}: when $\ph <p$ and $\ah <\alpha$, then the buyer can reduce the coverage once from $\alpha$ to $\ah$ by paying the fee $\gamma$ and a reduced premium $\ph$ thereafter.
    \item \emph{Cancellation Right}: as a special case of the step-down
        option with $\ph = \ah =0$, the resulting contract allows the buyer
        to terminate the default swap at time $\tau$.
\end{enumerate}
In addition,  the perpetual vanilla CDS corresponds to the case with $\gamma=0$, $p=\ph$ and $\alpha =\ah$, and the CDS spread is given by \eqref{def_p_pertpetual}. We ignore the  contract specifications with $(\ph - p)(\ah -\alpha)\leq 0$  since they would mean paying more (less) premium in exchange for a reduced (increased) protection after exercise.  In summary, we study the valuation of the (perpetual) \emph{American callable step-up/down} default swaps. For any fixed parameters $(p,\ph, \alpha, \ah, \gamma)$, the value $V(x)$ is referred to as the buyer's price, so the seller's price is $-V(x)$. The credit spread $p^*$ is determined from the equation $V(x;p^*,\ph, \alpha, \ah, \gamma) =0$ so that no cash transaction occurs at inception.

In preparation for our solution procedure, we first provide a useful representation of the buyer's value $V$. Define
\begin{align}
\acheck := \alpha - \ah \quad \textrm{and} \quad \pcheck := p - \ph. \label{acheck_pcheck}
\end{align}
Here,   $\acheck > 0$ and $\pcheck > 0$ hold for a step-down default swap and $\acheck < 0$ and $\pcheck < 0$ for a step-up default swap.
\begin{proposition}\label{prop-V} The perpetual American callable step-up/down default swap can be decomposed into a perpetual vanilla CDS plus a perpetual American payer/receiver swaption. Precisely, we have
\begin{align}V(x; p,\ph, \alpha, \ah, \gamma) = C(x;p,\alpha)+v(x;-\pcheck, -\acheck,\gamma),\label{valfn_lemma}\end{align}where  $C(\cdot)$ and $v(\cdot)$  are given in \eqref{Cxp} and  \eqref{Amer_v}, respectively.
\end{proposition}

\begin{proof}First, by a rearrangement of integrals, the expression inside the expectation in \eqref{valuefn} becomes
\begin{align*}
&1_{\{ \tau < \infty \}} \left(   \int_\tau^\theta e^{-rt} \pcheck\, \diff t - \int_0^\theta e^{-rt} p\, \diff t - e^{-r \tau} \gamma 1_{\{\tau < \theta\}} - e^{-r \theta} \acheck 1_{\{\tau < \theta\}}  + e^{-r \theta} \alpha\right) + 1_{\{ \tau = \infty \}} \left( - \int_0^\infty e^{-rt} p\, \diff t\right) \\
&= 1_{\{ \tau < \infty \}} \left(\int_\tau^\theta e^{-rt} \pcheck \,\diff t - e^{-r \tau} \gamma 1_{\{\tau < \theta\}} - e^{-r \theta} \acheck 1_{\{\tau < \theta\}} \right)  - \int_0^\theta e^{-rt} p\, \diff t + e^{-r \theta} \alpha
\end{align*}
since $\tau = \infty$ implies $\theta = \infty$ by the definition of $\mathcal{S}$. Because the last two terms do not depend on $\tau$, we can rewrite the buyer's value function as
\begin{align}V(x) &= \underbrace{\sup_{\tau \in \S} \E^x \left[1_{\{\tau < \infty \}} \left(\int_\tau^\theta e^{-rt}\pcheck \,\diff t - e^{-r \tau} \gamma  1_{\{\tau <\theta \}} -  e^{-r \theta} \acheck  1_{\{\tau <\theta\}} \right) \right]}_{=:f(x)}-\E^x \left[ \int_0^\theta e^{-rt} p\, \diff t \right]  +\alpha\, \E^x \left[ e^{-r \theta}  \right]. \label{V_x_decomposition}
\end{align}
Here, the last two terms in fact constitute $C(x;p,\alpha)$. Next, using the fact   $\{ \tau  < \theta,  \tau < \infty \} =\{ X_\tau >  0, \; \tau < \infty \}$ for every $\tau \in \mathcal{S}$ and the strong Markov property of
$X$ at time $\tau$, we rewrite the first term as
\begin{align}
{f(x)}= \sup_{\tau \in \S}\E^x \left[ e^{-r \tau} h(X_\tau) 1_{\{\tau < \infty\}}\right],\label{vx_inproof}
\end{align}
with
\begin{align}h(x) := 1_{\{x > 0 \}}\left( \E^x \left[ \int_0^\theta e^{-rt} \pcheck \, \diff t\,  - e^{-r \theta} \acheck \, \right] - \gamma \right)= 1_{\{x > 0 \}} \big(C(x;-\pcheck,-\acheck)- \gamma \big).\label{hx_inproof}\end{align}

Since $\theta\in\S$ and $h(X_\theta)=0$ a.s.\ on $\{\theta < \infty\}$, it follows from  \eqref{vx_inproof} that $f(x) \geq0$. Therefore, it is never optimal to exercise at any $\tau$ if $h(X_\tau)<0$. Consequently, we can replace $h(x)$ with $(h(x))^+$ in \eqref{vx_inproof}.  As a result,  with $-\pcheck, -\acheck  > (< )\, 0$, the function $f(x)$ is indeed the price of a perpetual \emph{American payer (receiver) default swaption} written on the buyer's (seller's) CDS price with strike $\gamma\geq 0$. This implies that $f(x) = v(x; - \pcheck, - \acheck, \gamma)$ $\forall x\in\R$, and therefore \eqref{valfn_lemma} follows.\end{proof}

The decomposition \eqref{valfn_lemma} in Proposition \ref{prop-V} yields a
\emph{static replication} of  the American callable step-up/down default swap. To
this end, one may also verify the result by a no-arbitrage argument. We
summarize the buyer's and seller's positions in the American callable
step-up/down default swaps in Table \ref{table:4kinds}.  As we shall discuss in Section  \ref{section_finite_maturity} below, this also holds for the finite-maturity case.

 As an  illustrative example, let us consider the step-up case where the
premium and protection are doubled after exercise, i.e. $\hat{p}=2p$ and
$\ah = 2 \alpha$. For any candidate exercise time $\tau$, the observable
market prevailing vanilla CDS spread  is given by $p(X_\tau;\alpha)$ in
\eqref{def_p_pertpetual}, and $C(X_\tau; p(X_\tau; \alpha), \alpha)  =0$ by
definition. Hence, if $p(X_\tau; \alpha) \le -\tilde{p} = p$ at $\tau$, then
$h(X_\tau) \le -\gamma \le 0$, and the buyer will not exercise. This is
intuitive because the buyer is better off giving up the step-up option and doubling his protection by entering a separate CDS at the lower market
spread $p(X_\tau; \alpha)$ at time $\tau$.

\subsection{The American Putable Step-Up and Step-Down Default Swaps} \label{sect-putableCDS} Applying  the ideas from the previous subsection, we  formulate the pricing problem for  the perpetual \emph{American putable step-up/down} default swaps. These default swaps allow the protection seller (and not the buyer)  to change the protection premium and default payment  for a fee  anytime prior to default.  Let $p$ and $\alpha$  be the initial premium and default payment. The seller may select a time $\tau$ to switch to a new premium $\ph$ and default payment $\ah$ for a switching fee $\gamma \geq 0$. The seller's maximal expected cash flow is
\begin{align}U&(x; p, \ph, \alpha, \ah, \gamma) \label{valuefn_put}\\
&:=\sup_{\tau \in \S} \E^x \left[ \int_0^\tau e^{-rt} p\,\diff t + 1_{\{\tau < \infty\}}\left( \int_\tau^\theta e^{-rt}\ph\,\diff t -e^{-r \tau} \gamma  1_{\{\tau <\theta \}} - e^{-r \theta}(\ah 1_{\{\tau <\theta\}} +\alpha 1_{\{\tau =\theta \}}) \right) \right]\notag  \\
&\,= \sup_{\tau \in \S} \E^x \left[1_{\{\tau < \infty \}} \left(-\int_\tau^\theta e^{-rt}\pcheck \,\diff t - e^{-r \tau} \gamma  1_{\{\tau <\theta \}} +  e^{-r \theta} \acheck  1_{\{\tau <\theta\}} \right) \right] \notag \\
&\,~\,~  +\E^x \left[ \int_0^\theta e^{-rt} p\, \diff t \right]  -\alpha\, \E^x \left[ e^{-r \theta}  \right]. \notag
\end{align}

In particular, we will study the American putable default swap   with a step-up option (i.e. $p<\ph$ and $\alpha <\ah$) or step-down option (i.e. $p>\ph$ and $\alpha >\ah$). Again, the credit spread $p^*$ is chosen so that the seller's value function  is zero, i.e.  $U(x;p^*,\ph, \alpha, \ah, \gamma) =0$.

Following the procedure in the proof of Proposition \ref{prop-V} or by a no-arbitrage argument, we can simplify the seller's value $U$ as follows:

\begin{proposition}\label{prop-U} The perpetual American putable step-up/down default swap  can be decomposed into a short perpetual vanilla CDS and a long perpetual American receiver/payer default swaption. Precisely, we have
\begin{align}U(x; p,\ph, \alpha, \ah, \gamma) = - C(x;p,\alpha)+u(x;-\pcheck,- \acheck,\gamma),\label{valfn_U}\end{align}where  $C(\cdot)$ and $u(\cdot)$  are given in \eqref{Cxp} and  \eqref{Amer_u}, respectively.
\end{proposition}

\begin{table}[t]\begin{small}
\begin{tabular}{|l||l|l|}
  \hline
 \emph{ Default Swap Types} & \emph{Protection Buyer's Position:} & \emph{Protection Seller's Position:} \\
 &$(+)$ a vanilla CDS and   &   $(-)$ a vanilla CDS and    \\
  \hline
   \emph{Callable Step-Up}   &  $(+)$ an American payer default swaption &  $(-)$ an  American payer default swaption  \\
   \hline
 \emph{Callable Step-Down}   & $(+)$ an American receiver default swaption &   $(-)$ an  American receiver default swaption  \\
  \hline
  \emph{Putable Step-Up}  & $(-)$ an  American receiver default swaption &  $(+)$ an American receiver default swaption  \\
   \hline
 \emph{Putable Step-Down}   & $(-)$ an American payer default swaption &   $(+)$ an American payer default swaption  \\
  \hline
\end{tabular}\end{small}\vspace{4pt}\caption{\small Positions of American callable/putable  step-up/down default swaps and their decompositions ($(+)/(-)$ stands for long/short). The seller's position is the opposite of the buyer's.  These decompositions hold for both the perpetual and finite-maturity cases (see  Propositions \ref{prop-V} and \ref{value_function_decomposition_finite} below). } \label{table:4kinds}
\end{table}

We summarize the buyer's and seller's positions in the American putable step-up/down default swaps in Table \ref{table:4kinds}.  To gain intuition on the seller's exercise decision, let us  look at the step-down  case where $\ph=0.5p$ and $\ah = 0.5 \alpha$. Recall that  $C(x; p, \alpha)$ is decreasing in $p$ and $C(X_\tau, p(X_\tau;\alpha), \alpha) = 0$ for any stopping time  $\tau$. If the market prevailing CDS spread is $p(X_\tau; \alpha) \le  p$ at some $\tau$, then the seller's default swaption payoff  is $-C(X_\tau; \pcheck, \acheck) -\gamma \leq -\gamma \leq 0$. The seller will not exercise  at $\tau$ since the protection of $0.5\alpha$ can be purchased from  a separate CDS at the lower prevailing spread $0.5 p(X_\tau;\alpha) \leq 0.5 p$.

\subsection{Symmetry Between Callable and Putable Default Swaps} \label{subsection_symmetry}By Propositions \ref{prop-V} and \ref{prop-U}, along with \eqref{equal_u_v}, we observe the following ``put-call parity" and symmetry identities:
\begin{align*}V(x;p, \ph, \alpha, \ah, \gamma ) - U(x;p, 2p-\ph, \alpha, 2\alpha-\ah, \gamma ) &= 2\, C(x;p, \alpha),\\
V(x;p, \ph, \alpha, \ah, \gamma ) + U(x;p, 2p-\ph, \alpha, 2 \alpha-\ah, \gamma )&= 2\,v(x;\ph-p, \ah-\alpha, \gamma ).
\end{align*}
The first equality means a long position in  an American callable step-up (step-down)  default swap and a short position in an American putable step-down (step-up)  default swap result in a double long position in a vanilla CDS. From the second equality,  a long position in both an American callable step-up (step-down)  default swap and  an American putable step-down (step-up)  default swap yields a double long position in an American payer (receiver) default swaption.  As we see in Section  \ref{section_finite_maturity}, this also holds for the finite-maturity case.

Furthermore, according to  \eqref{valfn_lemma} and \eqref{valfn_U}, the optimal exercise times for $V(x)$ and $U(x)$   are determined from $v(x)$ and $u(x)$ which depend on the triplet $(\pcheck, \acheck, \gamma)$ but not directly on $p$ and $\alpha$. Consequently, by \eqref{equal_u_v}, the same optimal exercising strategy applies for both
 \begin{enumerate}\item the protection buyer of an American callable default swap with a step-up (step-down) option with $(-\pcheck, -\acheck, \gamma)$, and
 \item the protection seller of an American putable default swap with a step-down (step-up)  option with $(\pcheck, \acheck, \gamma)$.
 \end{enumerate}
This observation means that it suffices to solve for two cases instead of four. Specifically, we shall solve for (i) the buyer's callable step-down case in (\ref{valfn_lemma}) and (ii) the seller's putable step-down case in (\ref{valfn_U}), both with $\pcheck > 0$ and $\acheck > 0$. In view of \eqref{vx_inproof} and the proof of Proposition \ref{prop-V}, this amounts to solving the following optimal stopping problems:
\begin{align}
v(x) &:= v(x; - \pcheck, - \acheck, \gamma) = \sup_{\tau \in \S}\E^x \left[ e^{-r \tau} h(X_\tau) 1_{\{\tau < \infty\}}\right], \label{ux1} \\
u(x) &:= u(x; - \pcheck, - \acheck, \gamma) = \sup_{\tau \in \S}\E^x \left[ e^{-r \tau} g(X_\tau) 1_{\{\tau < \infty\}}\right], \label{seller_problem}
\end{align}
where $\pcheck , \acheck > 0$ and
\begin{align}
h(x) &:= \left(\left( \frac \pcheck r - \gamma \right) - \left( \frac \pcheck r + \acheck \right) \zeta(x) \right) 1_{\{x > 0 \}}, \label{hx} \\
g(x) &:= \left( \left( -\frac \pcheck r - \gamma \right) + \left( \frac \pcheck r + \acheck \right) \zeta(x) \right) 1_{\{x > 0 \}}, \label{def_G}
\end{align}
for $x \in \R$.  Here,  $h(x)$ and $g(x)$ are computed using formula \eqref{Cxp}.

By inspecting (\ref{ux1}), it follows from \eqref{hx} that $h(x)\le 0$ $\,\forall x \in \R$ if $\gamma \ge  \pcheck/ {r}$. Financially, this means that the fee $\gamma$ to be paid exceeds the maximum benefit of stepping down, i.e. perpetual annuity with premium $p-\ph>0$. It is clear that choosing $\tau = \theta$ is optimal and the protection buyer will never exercise the step-down option.   Hence, we only need to study the non-trivial case with  the condition
\begin{align} 0 \leq \gamma < \frac \pcheck {r}.  \label{assumption_basic}\end{align}

For \eqref{seller_problem}, we have $g(x) \le 0$ $\,\forall x \in \R$ if  $g(0+) \leq 0$  because $g$ is decreasing in $x$ on $(0,\infty)$. Again, this means that $\theta$ is automatically optimal for the protection seller. Therefore, we shall focus on the case with $g(0+) > 0$  which also implies \begin{align} 0 \le \gamma <\acheck .\label{assumption_a}\end{align}
The intuition behind this is that the fee should not exceed the reduction in liability.

 .

\subsection{Solution Methods via Continuous and Smooth Fit}
We conclude this section by describing our solution procedure for the optimal stopping problems  under a general \lev model.  In the next section, we shall focus on the spectrally negative \lev model and derive an analytical solution.

For our first problem  (\ref{ux1}),   the protection buyer has an incentive to step-down when default  is less likely, or equivalently when $X$ is sufficiently high. Following this intuition, we denote the threshold strategy
\begin{align}
\tau_B^+ := \inf \left\{ t \geq 0: X_t \notin (0,B) \right\}, \quad B \geq 0. \label{definition_nu_b}
\end{align}Clearly, $\tau_B^+ \in \S$.
The corresponding expected payoff is given by
\begin{align}
v_B(x) := \E^x \left[ e^{-r \tau_B^+} h(X_{\tau_B^+}) 1_{\{\tau_B^+ < \infty\}}\right], \quad x \in \R.\label{vbx}
\end{align}
Note that $v_B(x) = h(x)=0$ for $x\leq0$. Sometimes it is more intuitive to consider the difference
\begin{align*}
\Delta_B(x) :=v_B(x)- h(x), \quad x \in \R.
\end{align*}

One common solution approach for many optimal stopping problems is \emph{continuous and smooth fit} (see \cite{Peskir_2001,Peskir_2006,Peskir_2000,Peskir_2002}). Applying to our problem, it involves two main steps:
\begin{enumerate}
\item[(a)] obtain $B^*$ that satisfies the continuous or smooth fit condition: $\Delta_{B^*}(B^*-) = 0$ or $\Delta_{B^*}'(B^*-) = 0$, and
 \item[(b)] verify the optimality of $\tau_{B^*}^+$ by showing (i) $v_{B^*}(x) \geq h(x)$ for $x \in \R$ and (ii) the process $M_t := e^{-r (t \wedge \theta)}v_{B^*}(X_{t \wedge \theta})$, $t \geq 0$, is a supermartingale.
\end{enumerate} To this end,  an analytical expression  for $v_{B}$ or  $\Delta_B$ would be useful.
\begin{lemma} \label{lemma_delta_b} Fix $B > 0$. The function $\Delta_B$ is given by
\begin{align}
\Delta_B(x)  =\begin{cases}\displaystyle \, \left( \frac \pcheck r - \gamma  \right)  \Lambda_1(x;B) + \left(\frac \pcheck r + \acheck \right) \Lambda_2(x;B)  + \gamma - \frac {\pcheck} r, & x \in (0,B),\\
\displaystyle \, 0, & x \notin  (0,B),\end{cases}\label{DeltaB}
\end{align}where $\Lambda_1(x;B):=\E^x\left[  e^{-r \tau_B^+}1_{\{\tau_B^+ < \theta, \; \tau_B^+ < \infty\}} \right]$ and  $\Lambda_2(x;B)  :=\E^x\left[  e^{-r \tau_B^+}1_{\{\tau_B^+ = \theta < \infty\}} \right]$.
\end{lemma}
As we shall see in Section \ref{sect-buyer-sol}, the functions $\Lambda_1(\cdot\,;B)$ and $\Lambda_2(\cdot\,;B)$ can be computed via the scale functions for a spectrally negative \lev model; see \eqref{property_scale_function} below.

In our second problem  (\ref{seller_problem}), the protection seller tends to  exercise the step-down option when default is likely, or equivalently when $X$ is sufficiently small.  Suppose the seller exercises at the first time  $X$ \emph{reaches or goes below} some fixed threshold $A \geq 0$; namely,
\begin{align*}
\tau_A^-:=\inf\{t\ge 0: X_t \leq A\}.
\end{align*}
Then, the corresponding expected payoff is given by
\begin{align*}
u_A(x) := \E^x \left[ e^{-r \tau_A^-} g(X_{\tau_A^-}) 1_{\{\tau_A^- < \infty\}}\right], \quad x \in \R.
\end{align*}
Again, we denote the difference between continuation and exercise by $\Delta_A(x) := u_A(x)- g(x)$ for $x \in \R$.

For this problem, the continuous and smooth fit solution approach is to
\begin{enumerate}
\item[(a)] obtain $A^*$ that satisfies the continuous or smooth fit condition: $\Delta_{A^*}(A^*+) = 0$ or $\Delta_{A^*}'(A^*+) = 0$, and
\item[(b)] verify the optimality of $\tau_{A^*}^-$ by showing (i) $u_{A^*}(x) \geq g(x)$  for $x \in \R$ and (ii) the process $\widetilde{M}_t := e^{-r (t \wedge \theta)}u_{A^*}(X_{t \wedge \theta})$, $t \geq 0$, is a supermartingale.
\end{enumerate}
This  method requires some expression for $\Delta_A$, which is summarized as follows:
\begin{lemma} \label{lemma_delta_A}
Fix $A > 0$. The function $\Delta_A$ is given by\begin{align} \label{delta_seller}
\Delta_A(x) =\begin{cases} \displaystyle\, \left( \gamma + \frac \pcheck r \right) \left( 1-\zeta(x-A) \right) - (\acheck-\gamma)\, \Gamma(x;A),&  x > A,\\
\displaystyle\,  0, & x \leq A, \end{cases}
\end{align}
where
\begin{align} \label{def_gamma}
\Gamma(x;A) := \E^x \left[ e^{-r \tau_A^-} 1_{\{X_{\tau_A^-} < 0, \, \tau_A^- < \infty \}}\right].
\end{align}
\end{lemma}
The function $\Gamma(\cdot\,; A)$ and Laplace transform $\zeta(\cdot)$  can be also expressed in terms of the scale function for a spectrally negative \lev model; see \eqref{laplace_theta} and Lemma \ref{remark_gamma_x_a} below.

\section{Solution Methods under the Spectrally Negative \lev Model}\label{sect-buyer-sol}
We proceed to solve the optimal stopping problems $v(x)$ and $u(x)$ in  (\ref{ux1}) and (\ref{seller_problem}) for spectrally negative \lev processes. Our main results are Theorems \ref{optimality_buyer} and \ref{proposition_optimality_case_1} which provide the optimal solutions  for $v(x)$ and $u(x)$, respectively. In turn, the American callable/putable step-up/down default swap can be immediately priced in view of Propositions \ref{prop-V} and \ref{prop-U}.

\subsection{The Spectrally Negative \lev Process and  Scale Function}
Let $X$ be a spectrally negative \lev process with  the \emph{Laplace exponent}
\begin{align}
\psi(s) := \log \E^0 \left[ e^{s X_1} \right] =c s +\frac{1}{2}\sigma^2 s^2 +\int_{(
0,\infty)}(e^{-s x}-1+s x 1_{\{0 <x<1\}})\,\Pi(\diff x), \quad  {s \in \mathbb{C}},  \label{laplace_spectrally_negative}
\end{align}
where $c \in\R$, $\sigma\geq 0$ is called the Gaussian coefficient, and $\Pi$ is  a measure on $\R$ such that $\Pi(-\infty,0]=0$ and
\begin{align*}
\int_{(0,\infty)} (1  \wedge x^2) \Pi( \diff x)<\infty.
\end{align*}
See, e.g. Theorem 1.6 of \cite{Kyprianou_2006}.  The risk neutral condition requires that $\psi(1)=r$ so that the discounted value of the reference entity is a $\p$-martingale. By Lemma 2.12 of \cite{Kyprianou_2006}, if further we have
\begin{align}
\int_{(0,\infty)} (1 \wedge x)\, \Pi(\diff x) < \infty,  \label{cond_bounded_variation}
\end{align}
then  the Laplace exponent can be expressed as
\begin{align*}
\psi(s) =\mu  s+  \frac{1}{2}\sigma^2 s^2 + \int_{(
0,\infty)}(e^{-s x}-1)\,\Pi(\diff x), \quad s \in \mathbb{C},
\end{align*}
where $\mu := c + \int_{(0,1)}x\, \Pi(\diff x)$. Recall that the process has paths of bounded variation if and only if $\sigma = 0$ and \eqref{cond_bounded_variation} holds. A special example is a compound Poisson process with $\Pi(\R)=\lambda$, where $\lambda$ is the finite rate of jumps. We ignore the  negative subordinator case ($X$ decreasing a.s.). This means that we require $\mu$ to be strictly positive when $X$ is of bounded variation.

By Theorem 8.1 of \cite{Kyprianou_2006}, for any spectrally negative \lev process, there exists an \emph{(r-)scale
function} $W^{(r)}: \R \mapsto \R$, $r\ge 0$ such that $W^{(r)}(x)=0$ on $(-\infty,0)$, and is characterized on $[0,\infty)$ by the Laplace transform:
\begin{align}\label{eq:scale}
\int_0^\infty e^{-s x} W^{(r)}(x) \diff x = \frac 1
{\psi(s)-r}, \qquad s > \lapinv,
\end{align}
where $\lapinv :=\sup\{\lambda \geq 0: \psi(\lambda)=r\}$.

The properties of the scale function  \cite[Theorem 8.1]{Kyprianou_2006}  allow us to derive the analytic formulas for $\Lambda_1(x;B)$ and $\Lambda_2(x;B)$ from Lemma \ref{lemma_delta_b}. Precisely, for $0 < x < B$,
\begin{align}
\begin{split}
\Lambda_1(x;B) &= \E^x \left[ e^{-r \tau_B^+} 1_{\{\tau_B^+ < \theta, \, \tau_B^+ < \infty \}}\right] = \frac {W^{(r)}(x)} {W^{(r)}(B)}, \\
\Lambda_2(x;B) &= \E^x \left[ e^{-r \tau_B^+} 1_{\{\tau_B^+ = \theta  < \infty \}}\right] = Z^{(r)} (x) - Z^{(r)} (B) \frac {W^{(r)}(x)} {W^{(r)}(B)},
\end{split} \label{property_scale_function}
\end{align}
where $Z^{(r)} (x) := 1 + r \int_0^x W^{(r)} (y) \diff y$, $x \in \R$. Notice that $Z^{(r)}(x) = 1$ for   $x \in (-\infty,0]$. The Laplace transform of $\theta$ in (\ref{zeta}) is given by
\begin{align}
\zeta(x) = Z^{(r)}(x) - \frac r {\lapinv} W^{(r)}(x), \quad x \in \R \backslash \{0\}. \label{laplace_theta}
\end{align}

Henceforth, we assume that $\Pi$ does not have atoms, which  guarantees that $W^{(r)}$ is $C^1$ on $(0,\infty)$ (see  \cite{Chan_2009}).  Moreover, as in (8.18) of \cite{Kyprianou_2006},
\begin{align}
 \frac {W^{(r)'}(y)} {W^{(r)}(y)} \leq \frac {W^{(r)'}(x)} {W^{(r)}(x)},  \quad y > x > 0. \label{assumeW}
\end{align}
From Lemmas 4.3 and 4.4 of \cite{Kyprianou_Surya_2007}, we also summarize the behavior in the neighborhood of zero.
\begin{lemma} \label{lemma_zero}
For every $r \geq 0$, we have
\begin{align*}
W^{(r)} (0) = \left\{ \begin{array}{ll} 0, & \textrm{unbounded variation} \\ \frac 1 {\mu}, & \textrm{bounded variation} \end{array} \right\} \quad \textrm{and} \quad
W^{(r)'} (0+) = \left\{ \begin{array}{ll}  \frac 2 {\sigma^2}, & \sigma > 0 \\   \infty, & \sigma = 0 \; \textrm{and} \; \Pi(0,\infty) = \infty \\ \frac {r + \Pi(0,\infty)} {\mu^2}, & \textrm{compound Poisson} \end{array} \right\}.
\end{align*}
\end{lemma}

\subsection{Callable Step-Down Default Swap} \label{subsection_callable_step_down}
We proceed to solve for $v(x)$ in (\ref{ux1}) for the callable step-down default swap. First, we consider the expected payoff function $v_B(x)$ in \eqref{vbx} with some threshold $B$:
\begin{align}\label{v_cand2}
v_{B} (x) = \left\{ \begin{array}{ll} h(x),   & x \in [B,\infty),\\ h(x) + \Delta_{B} (x), & x \in (0,B), \\ 0, & x \in (-\infty,0],\end{array}\right.
\end{align} for $0 \leq B < \infty$. If $B = \infty$, then $v_{B}(x)=0$, $x\in \R$. If $B = 0$, then $v_{B}(x)= h(x)$, $x \in \R$.
Applying  \eqref{property_scale_function} and \eqref{laplace_theta} and Lemma \ref{lemma_delta_b} to the stopping value  $h(x)$ and difference function $\Delta_B(x)$, we can express them in terms of the scale function, namely,
\begin{align}
h(x) &=  \left[ \pcheck \left( \frac 1 r \left( 1- Z^{(r)}(x) \right) + \frac 1 {\lapinv} W^{(r)}(x) \right)   -  \acheck \left(Z^{(r)}(x) - \frac r {\lapinv} W^{(r)}(x) \right) - \gamma \right] 1_{\{x > 0\}}, \label{hx_scale}\\
\Delta_B(x) &=  \left[\frac \pcheck r \left( Z^{(r)} (x) - 1 \right) + \acheck Z^{(r)} (x) - \frac {W^{(r)}(x)} {W^{(r)}(B)} G^{(r)}(B)+ \gamma \right] 1_{\{0 < x < B\}},\label{delta_scale}
\end{align}
where
\begin{align}
\label{GrB}G^{(r)}(B) := \frac \pcheck r  \left( Z^{(r)} (B) - 1 \right) + \acheck Z^{(r)} (B) + \gamma, \quad B \geq 0.
\end{align}

\begin{remark} \label{remark_continuity_v_B}
From \eqref{delta_scale}, we observe that $\Delta_B(B-) = 0$. This implies that  \textbf{continuous fit} $v_B(B-) = v_B(B)$ must hold  for all $B > 0$.
\end{remark}

To obtain  the candidate optimal  threshold, we consider the  \textbf{smooth fit condition} $\Delta_B'(B-)=0$. To this end, we compute from \eqref{delta_scale} the derivatives
\begin{align}
\label{varrho1}\varrho(B) &:= \Delta_B'(B-)  =  \left( \pcheck + \acheck r \right) W^{(r)} (B) - \frac {W^{(r)'}(B)} {W^{(r)}(B)} G^{(r)}(B), \\
\label{varrho2}\varrho'(B) &=  - \left(\frac \partial {\partial B} \frac {W^{(r)'}(B)} {W^{(r)}(B)} \right) G^{(r)}(B), \quad B>0.
\end{align}
Here $\varrho(B)$ is continuous on $(0,\infty)$ and  \eqref{varrho2} holds at which the second derivative of $W^{(r)}(B)$ exists (which holds for Lebesgue-a.e.\ $B > 0$).

Observing from \eqref{GrB}  that $G^{(r)}(B) \geq \acheck + \gamma > 0$ for $B \geq 0$ and by \eqref{assumeW},  we deduce that   $\varrho(B)$ is   increasing   in $B$. Therefore, there exists  \emph{at most} one $B^*\in(0,\infty)$ satisfying the smooth fit condition, which by \eqref{varrho1} is equivalent to
\begin{align}
\varrho(B^*) =0.
\label{definition_b_star}
\end{align}
If it exists, then this is our candidate optimal threshold, and $v_{B^*}(x)$ is the candidate value function for \eqref{ux1}.

The smooth fit condition fails if \,(a) $\varrho(B) \geq 0$ \,$\forall B > 0$, or\, (b) $\varrho(B) < 0$\, $\forall B > 0$. Under each of these scenarios, we need another way to deduce the candidate optimal threshold. To this end, let us consider the derivative of $v_B(x)$ with respect to $B$. For $0 < x < B$,
\begin{align}
\frac \partial {\partial B} v_B(x) = \frac \partial {\partial B} \Delta_B(x)= -\frac {W^{(r)}(x)} {W^{(r)}(B)}\left[ \left( r \acheck + \pcheck \right) W^{(r)} (B) - \frac {W^{(r)'}(B)} {W^{(r)}(B)} G^{(r)}(B) \right] =  -\frac {W^{(r)}(x)} {W^{(r)}(B)} \varrho(B).\label{vb_derv}
\end{align}

Under scenario (a),  $\varrho(B) \geq 0$ in \eqref{vb_derv} implies that  $v_B(x)$ is decreasing in $B$ for any $x < B$, so we choose $B^*=0$ as our candidate optimal threshold. In this case, the buyer will stop immediately ($\tau_{0}^+ = 0$), and the corresponding expected payoff is
$v_{B^*}(x)=h(x)$ (see \eqref{vbx}). As we show next, $B^*=0$ is possible only when $X$ is of bounded variation.
\begin{lemma} \label{lemma_b_star_zero}   We have $B^* = 0$ if and only if   $\sigma = 0$ and   $\pcheck - r \gamma -  (\acheck + \gamma)\Pi(0,\infty) \geq 0$.
\end{lemma}

As for scenario (b),  it follows from \eqref{vb_derv} that $v_B(x)$ is increasing in $B$ for any $x < B$. Therefore, we set $B^*=\infty$, meaning that the buyer will never exercise ($\tau_{\infty}^+ = \theta$), and the corresponding expected payoff is $v_{\infty}(x)=0$ (see \eqref{vbx}). In fact, this   corresponds to the   case where the payoff $h(x) \leq 0$\,$\forall x\in \R$. To see this, we deduce from \eqref{vb_derv} and continuous fit in Remark \ref{remark_continuity_v_B} that, for any arbitrarily fixed $x > 0$,  $\liminf_{B \rightarrow \infty}v_B (x) \geq v_{x}(x) = h(x)$. Then, applying  Fatou's lemma and because $\tau_B^+ \xrightarrow{B \uparrow \infty} \theta$ a.s., we obtain
$\limsup_{B \rightarrow \infty}v_B (x) \leq 0$, and thus $0\geq h(x)$.  Therefore, under   condition \eqref{assumption_basic}, we have already \emph{ruled out} scenario (b).

\begin{figure}[htbp]
\begin{center}
\centering
 \includegraphics[scale=0.55]{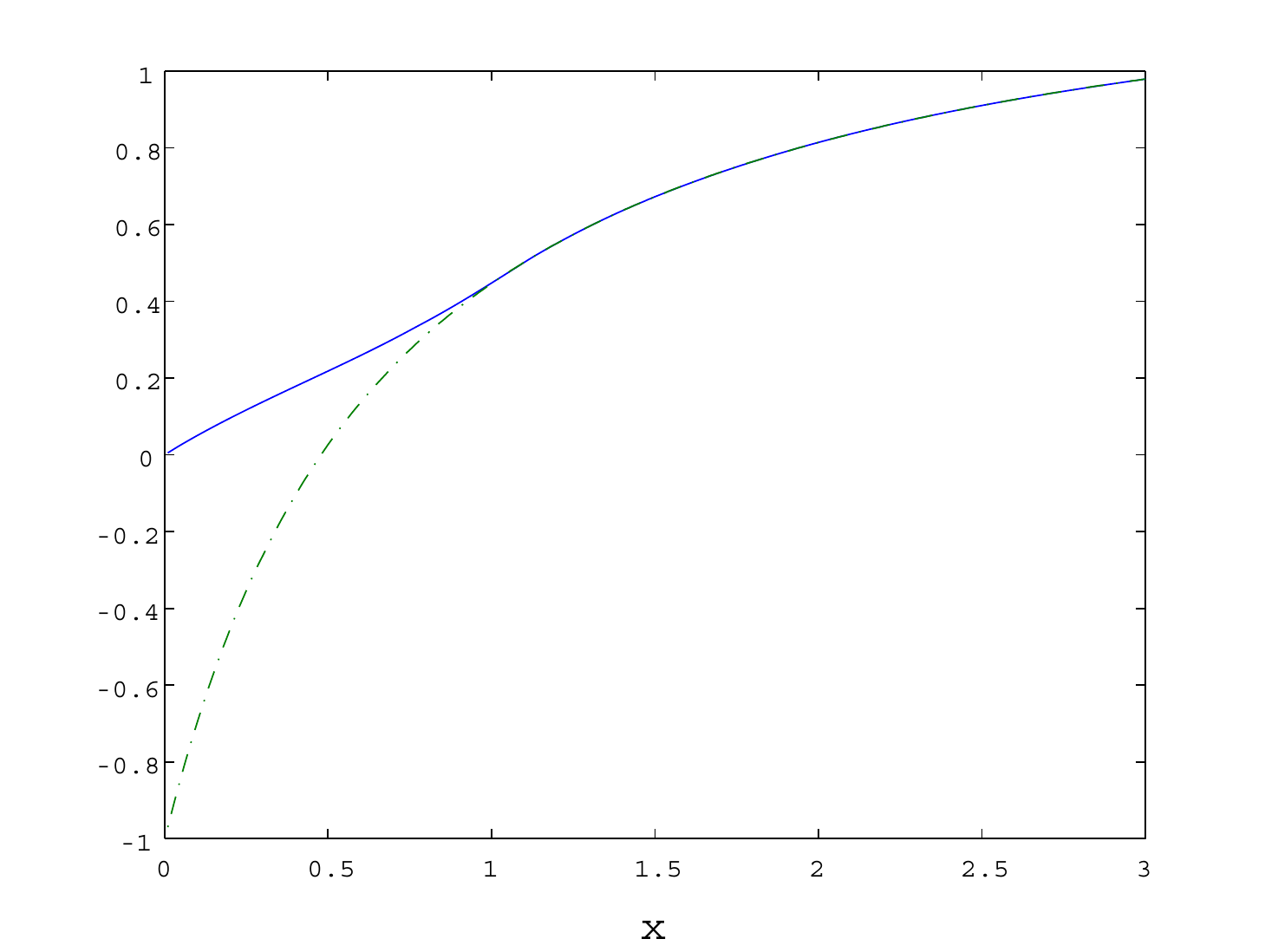} \includegraphics[scale=0.55]{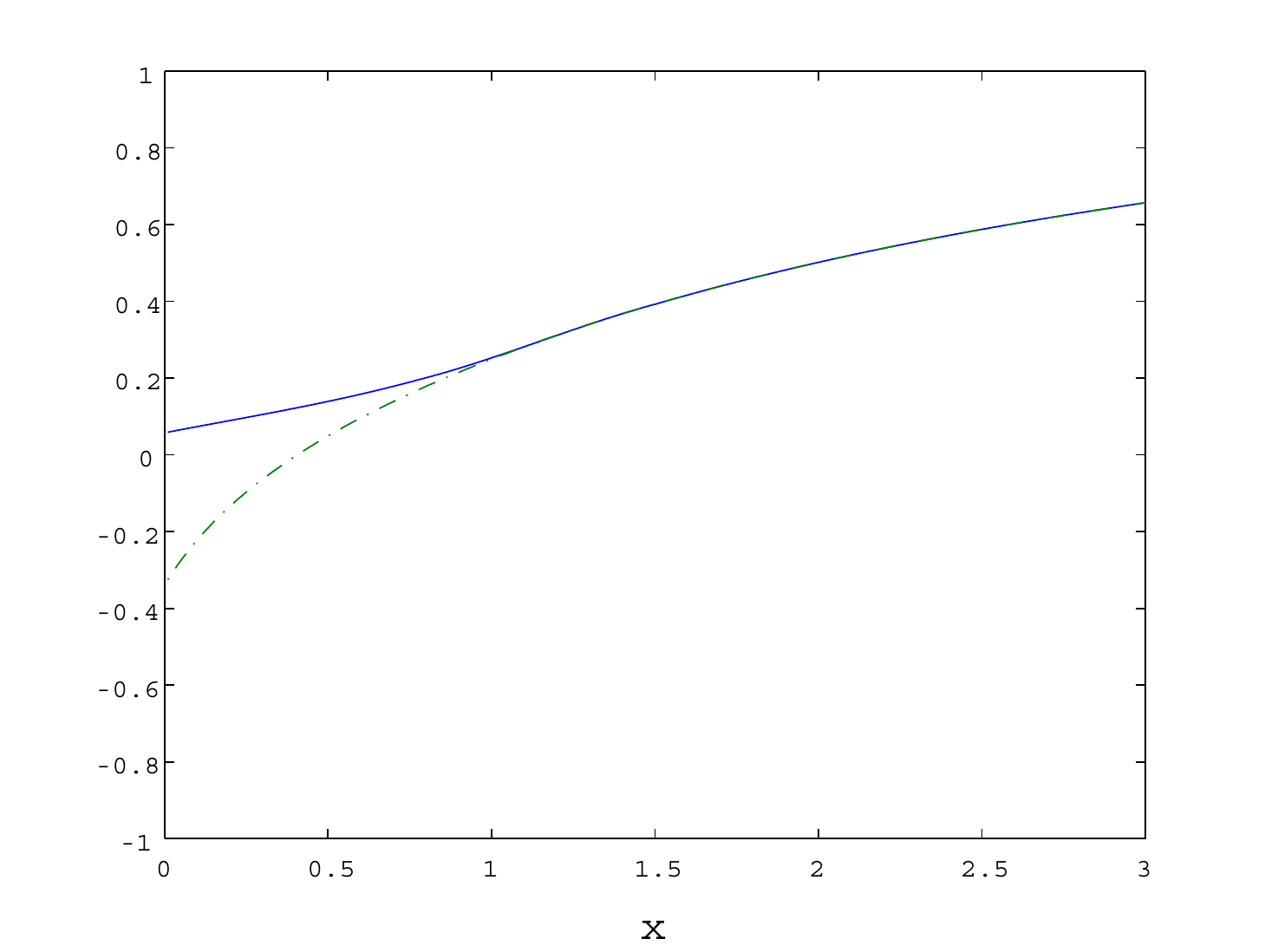}
\caption{\small{Continuous and smooth fits of the value function $v_{B^*}(\cdot)$ (solid curve) and stopping value $h(\cdot)$ (dashed curve). Notice that $v_{B^*}(0+)=0$  in the unbounded variation case (left) while  $v_{B^*}(0+)>0$ in the bounded variation case (right). }}\label{plot_v_h}
\end{center}
\end{figure}

To summarize,   we propose the  candidate value function $v_{B^*}(x)$ corresponding to the candidate optimal threshold $B^*$  from \eqref{definition_b_star} for $0 < B^* < \infty$, or $B^*=0$ otherwise. By direct computation using \eqref{v_cand2}-\eqref{delta_scale}, the value function is
\begin{align}
v_{B^*}(x)
= \left\{ \begin{array}{ll}\displaystyle  W^{(r)}(x) \left( \frac {\pcheck + \acheck r} {\Phi(r)}  - \frac {G^{(r)}(B^*)} {W^{(r)}(B^*)} \right) 1_{\{x \neq 0\}}, & -\infty < x < B^*, \\ h(x), & x \geq B^*. \end{array} \right. \label{U-smaller_than_b_star}
\end{align}
Here, on $(-\infty,0)$, $v_{B^*}(x) = 0$ because $W^{(r)}(x)=0$.

The next step is to verify the optimality of $v_{B^*}$. We shall show that (i) $v_{B^*}$ dominates $h$ and (ii) the stochastic process
\begin{align}
M_t := e^{-r (t \wedge \theta)}v_{B^*}(X_{t \wedge \theta}), \quad t \geq 0 \label{definition_m_buyer}
\end{align}
is a supermartingale.

We address the first part as follows. When $B^*=0$, it follows  from \eqref{vbx} or the arguments above that $v_{B^*}(x)=h(x)$. When $B^*\in(0,\infty)$,   $\varrho(B)$ is monotonically increasing and attains $0$ at $B^*$. Therefore, $v_{B}(x)$  is increasing in $B$ for  $B \in [x, B^*]$ by \eqref{vb_derv}. Then, by  continuous fit   in Remark \ref{remark_continuity_v_B}, for any arbitrarily fixed $x$, taking $B=x$, we have $v_{B^*}(x) \geq v_x (x) = h(x)$.  Moreover, $v_{B^*}(x) = h(x)$ on $x \leq B^*$ by definition. Hence, we conclude that:
\begin{align}v_{B^*}(x) \geq  h(x), \quad \text{ for } x \in \mathbb{R}.\label{v_greater_than_h}
\end{align}

We now pursue the supermartingale property of the process $M$ in (\ref{definition_m_buyer}).  Define the generator $\mathcal{L}$ of $X$ by
\begin{align*}
\mathcal{L} f(x) = c f'(x) + \frac 1 2 \sigma^2 f''(x) + \int_0^\infty \left[ f(x-z) - f(x) +  f'(x) z 1_{\{0 < z < 1\}} \right] \Pi(\diff z)
\end{align*}
for the unbounded variation case and
\begin{align*}
\mathcal{L} f(x) = \mu f'(x) + \int_0^\infty \left[ f(x-z) - f(x) \right] \Pi(\diff z)
\end{align*}
for the bounded variation case. The supermartingale property of $M$ is due to the following lemma (see the Appendix for a proof):

\begin{lemma} \label{generator_non_positive}
The function $v_{B^*}(\cdot)$ satisfies
\begin{align*}
(\mathcal{L}-r) v_{B^*}(x) \leq 0, \quad x \in (0,\infty) \backslash \{B^* \}.
\end{align*}
\end{lemma}

In summary,  inequality \eqref{v_greater_than_h} and Lemma \ref{generator_non_positive} prove the optimality of $v_{B^*} (x)$, leading to  our  main result:
\begin{theorem} \label{optimality_buyer}The candidate function $v_{B^*} (x)$, with threshold $B^*$ given by \eqref{definition_b_star}, is optimal for  \eqref{ux1}. Precisely,
\[v_{B^*}(x) = v(x) = \sup_{\tau \in \mathcal{S}} \E^x \left[ e^{-r \tau} h(X_\tau) 1_{\{\tau < \infty\}}\right],
\]
and $\tau_{B^*}^+$ is the optimal stopping time.
\end{theorem}
The proofs and related technical details are provided  in the Appendix. As a result, we have solved for the callable step-down and putable step-up default swaps in view of the decompositions by Propositions \ref{prop-V} and \ref{prop-U}.

\subsection{Putable Step-Down Default Swap}\label{sect-seller-sol}
We now turn our attention to the putable step-down default swap, which amounts to solving $u(x)$ in (\ref{seller_problem}) with $\pcheck > 0$ and $\acheck > 0$.  Under the spectrally negative \lev model, the function $\Gamma(\cdot \,;A)$ in \eqref{def_gamma} can be expressed explicitly by scale functions as shown in Lemma \ref{remark_gamma_x_a} below. This together with (\ref{laplace_theta}) expresses explicitly $\Delta_A(\cdot)$ in (\ref{delta_seller}), which is needed to apply continuous and smooth fit principle.

Define
\begin{align}
\rho(A) := \int_A^\infty \Pi (\diff u)  \left( 1 -  e^{-\lapinv (u-A)} \right), \quad A > 0. \label{varrho_A}
\end{align}
It is clear that   $\rho(A)$ decreases monotonically in $A$. Next, we express $\Gamma(\cdot\,;A)$ using  $\rho(A)$ and the scale function. For its proof, we refer to  Lemma 4.3 of \cite{Egami_Yamazaki_2010}.

\begin{lemma} \label{remark_gamma_x_a} Fix $A > 0$, we have
\begin{align*}
\Gamma(x\,;A) =
\frac 1 {\lapinv} W^{(r)}(x-A) \rho(A) - \frac 1 r \int_A^\infty \Pi (\diff u)  \left( Z^{(r)}(x-A) - Z^{(r)} (x-u) \right), \quad \text{for } x \geq A,
\end{align*}and\, $\Gamma(x\,;A) =0$\, for $x<A$.
\end{lemma}

We proceed to consider the continuous fit condition.  First, it follows from  (\ref{delta_seller}) that, for every $A > 0$,
\begin{align} \label{delta_A_A_seller}
\begin{split}
\Delta_A(A+)
&= \left( \gamma + \frac \pcheck r \right) \left( \frac r {\lapinv} W^{(r)}(0)\right) - (\acheck - \gamma) \left( \frac 1 {\lapinv} W^{(r)}(0) \rho(A)  \right) \\
&= \frac {W^{(r)}(0)} {\lapinv}  \left( \left( r \gamma + \pcheck \right)  - (\acheck - \gamma)  \rho(A)  \right).
\end{split}
\end{align}
If $X$ is of unbounded variation, then $W^{(r)}(0)=0$ by Lemma \ref{lemma_zero}, and therefore continuous fit holds for every $A > 0$. Nevertheless, for the bounded variation case,   we can apply the \textbf{continuous fit condition}: $\Delta_A(A+) =0$, which is equivalent to
\begin{align}
(\acheck - \gamma) \rho(A) = \gamma r + \pcheck. \label{smooth_fit_seller}
\end{align}

For the unbounded variation case, we apply the smooth fit condition. By differentiation, we have
\begin{align*}
\left. \frac \partial {\partial x} \zeta(x-A) \right|_{x=A+} = - \frac r {\lapinv} W^{(r)'}(0+)  \quad \textrm{and} \quad \left. \frac \partial {\partial x} \Gamma(x;A) \right|_{x = A+} = \frac {\rho(A)} {\lapinv} W^{(r)'}(0+),
\end{align*}
where $W^{(r)'}(0+) > 0$ (in particular, $W^{(r)'}(0+) = \infty$ if $\sigma = 0$ and $\Pi(0,\infty)=\infty$  by Lemma \ref{lemma_zero}). Therefore,
\begin{align*}
\Delta'_A(A+)  = \frac {W^{(r)'}(0+)} {\lapinv} \left( (\gamma r + \pcheck) - (\acheck - \gamma) \rho(A) \right).
\end{align*}
Consequently, the\textbf{ smooth fit condition}, $\Delta'_A(A+) =0$, is also equivalent to  (\ref{smooth_fit_seller}).

In summary, we look for the solution  to  (\ref{smooth_fit_seller}), denoted by $A^*$, which will be our candidate optimal threshold. Since $\rho(A)$ is monotonically decreasing, there exists \emph{at most} one $A^*$ that satisfies (\ref{smooth_fit_seller}). If it does not exist, we set the threshold $A^* = 0$.

When $X$ has paths of bounded variation,  we must have $A^* > 0$ under assumption \eqref{assumption_a}. Indeed, if $A^* = 0$, then it follows from  (\ref{delta_A_A_seller}) that $\Delta_A(A+) > 0$ for every $A > 0$. This implies that there exists $\varepsilon > 0$ such that $u_\varepsilon (\varepsilon+) > g(0+)$. However, since $g(0+)$ attains the global maximum (because $g(0+) > 0$ by \eqref{assumption_a}), this is a contradiction.  Hence $A^*=0$ is impossible when  $X$ is of bounded variation.

In the case with $A^* > 0$,  we take $A^*$ to be our candidate optimal threshold, and the corresponding stopping time is $\tau_{A^*}^-$. The  candidate value function is given by
\begin{align*}
u_{A^*}(x) = \left\{ \begin{array}{ll} g(x) + \Delta_{A^*}(x), &  x > A^*, \\ g(x), & 0 < x \leq A^*, \\
0, & x \leq 0.\end{array} \right.
\end{align*}
For $x >0$, we can apply     (\ref{smooth_fit_seller}) to express it as
\begin{align}\label{u_tilde_a}
u_{A^*}(x) = (\acheck- \gamma) \left(\frac 1 r \int_{A^*}^\infty \Pi (\diff u)   \left[ Z^{(r)} (x-A^*) - Z^{(r)}(x-u) \right]  \right)
- \left( \gamma + \frac \pcheck r \right) Z^{(r)}(x-A^*) + \left( \frac \pcheck r + \acheck \right) \zeta(x).
\end{align}

When $A^* = 0$, we consider the candidate value function defined by (see \eqref{def_G} and \eqref{delta_seller})\begin{align}
\begin{split}
u_{A^*}(x) &:= \lim_{A \downarrow 0} u_A(x)\\
& = -(\acheck - \gamma) \Gamma(x;0)
- \left( \gamma + \frac \pcheck r \right) \zeta(x) + \left( \frac \pcheck r + \acheck \right) \zeta(x)  \\
&=(\acheck - \gamma) (\zeta(x) - \Gamma(x;0))
\end{split} \label{def_w2}
\end{align}
for all $x > 0$ and $u_{A^*}(x)=0$ for every $x \leq 0$.  Here $\Gamma(x;0) := \lim_{A \downarrow 0}\Gamma(x;A)$ is well-defined according to Lemma \ref{remark_gamma_x_a} since $\rho(0) := \lim_{A \downarrow 0}\rho(A) < \infty$ if $A^* = 0$; see \eqref{smooth_fit_seller}.

This implies a strategy by which the seller will delay until $X$ is arbitrarily close to zero, and exercise at a  sufficiently small level $\varepsilon > 0$.  This can be realized by monitoring  $X$ as it \emph{creeps downward} through zero (see Section 5.3 of \cite{Kyprianou_2006}). The seller may lose the opportunity to exercise prior to default if $X$ suddenly jumps across (below) zero.  In fact, $\zeta(x) - \Gamma(x;0) = \E^x \left[ e^{-r\theta}\right] - \E^x \left[ e^{-r\theta} 1_{\{X_\theta < 0, \, \theta < \infty\}}\right] = \E^x \left[ e^{-r\theta} 1_{\{X_\theta = 0, \, \theta < \infty \}}\right] \geq 0$ and this is strictly positive if and only if $\sigma > 0$ (see Exercise 7.6 of \cite{Kyprianou_2006}).  The optimality result below implies that $A^* = 0$ can happen only when $X$ has a diffusion component by our assumption $g(0+) > 0$.

\begin{figure}[htbp]
\begin{center}
 \centering
 \includegraphics[scale=0.55]{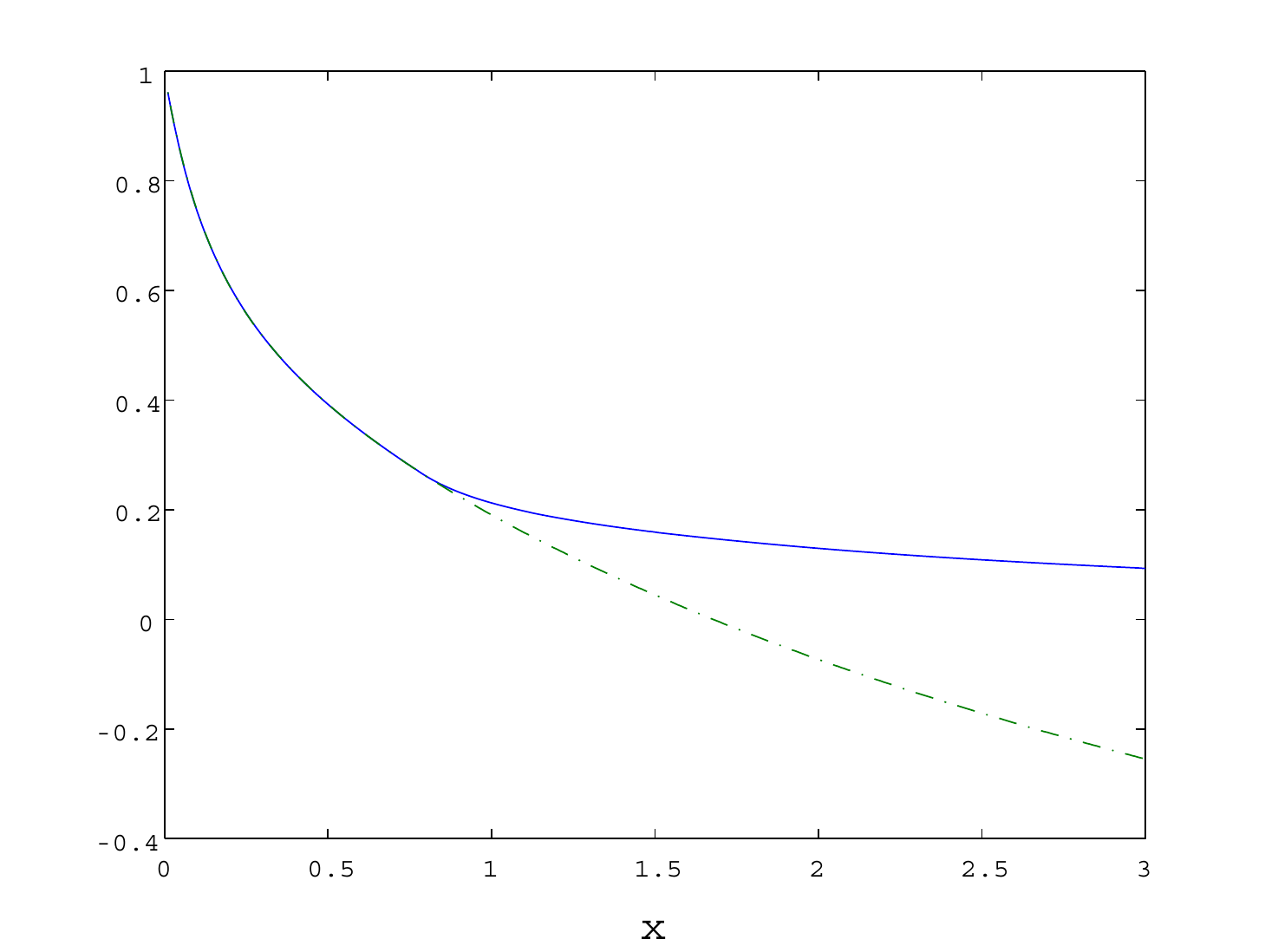}  \includegraphics[scale=0.55]{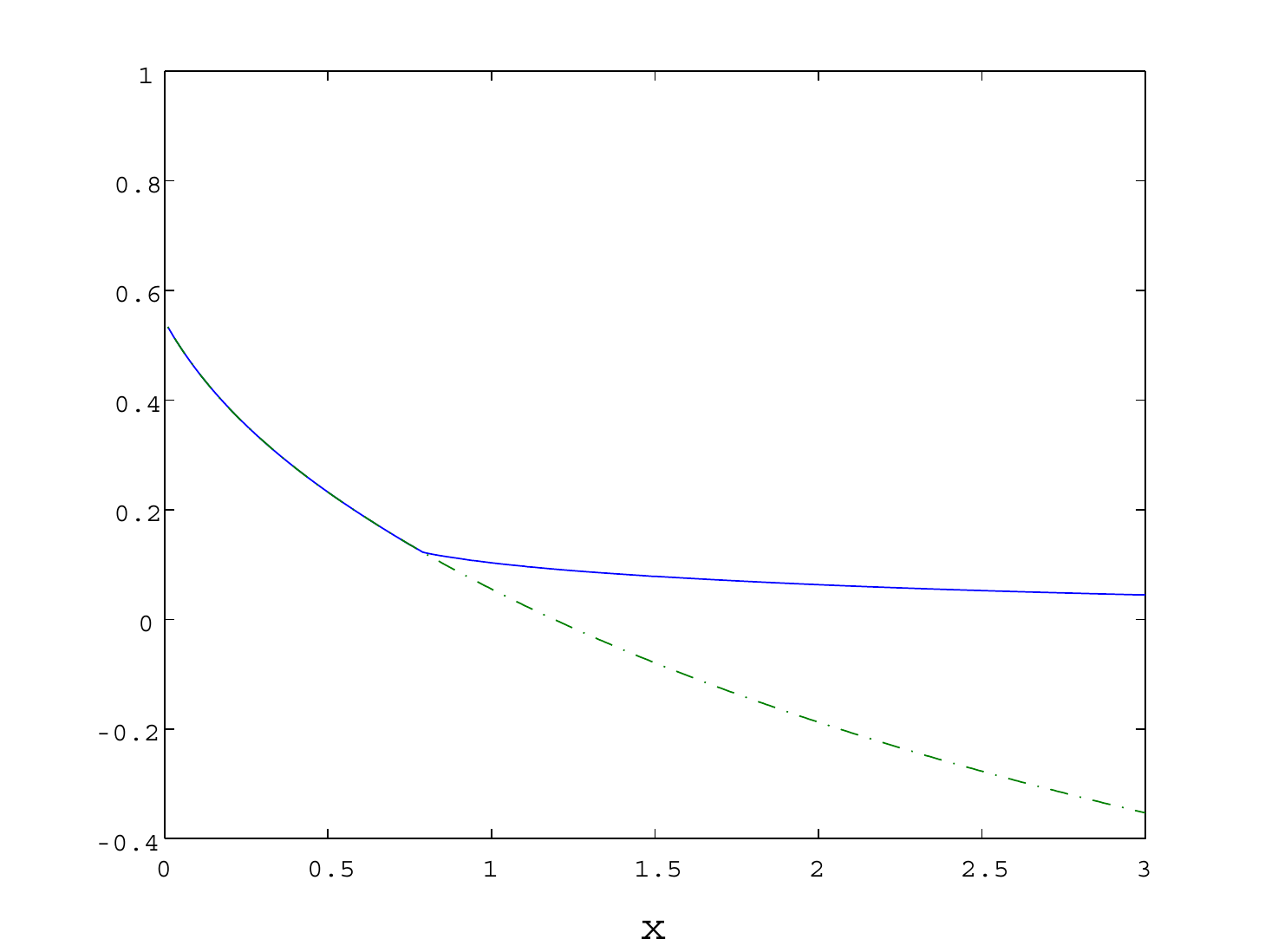}
\caption{\small{Continuous and smooth fits of the value function $u_{A^*}(\cdot)$ (solid curve) and stopping value $g(\cdot)$ (dashed curve). Notice that $u_{A^*}(\cdot)$ is $C^1$ at $A^*$ in the unbounded variation case (left) while it is $C^0$ in the bounded variation case (right).}}\label{plot_u_g}
\end{center}
\end{figure}

For optimality verification, we shall show\, (i) $u_{A^*} \geq g$ and\, (ii)  $e^{-r (t \wedge \theta)}u_{A^*}(X_{t \wedge \theta})$, $t \geq 0$, is a supermartingale. We shall first show (i) via:

\begin{lemma} \label{u_minimum_seller}
For every $0 < A < x$, we have
\begin{align}\label{signofderiv}
\frac \partial {\partial A} u_A(x) = \frac \partial {\partial A} \Delta_A(x) < 0 \quad \Longleftrightarrow \quad (\acheck-\gamma) \rho(A) - (\gamma r + \pcheck) < 0  \quad \Longleftrightarrow \quad A > A^*.
\end{align}
\end{lemma}

Suppose $A^*>0$. For the unbounded variation case, we   apply Lemma \ref{u_minimum_seller} for  any fixed $x \geq A^*$, along with the continuous fit, to obtain the inequality: $u_{A^*}(x) \geq u_x (x) = g(x)$, $x \geq A^*$. When $X$ is of bounded variation, we observe from  (\ref{delta_A_A_seller}) that  $\Delta_x(x+) \geq 0$ for $x \geq A^*$, which implies that
\begin{align}
u_{A^*}(x) \geq \lim_{A \uparrow x}u_A(x) = g(x) + \lim_{A \uparrow x} \Delta_A(x) \geq g(x),  \quad x \geq A^*. \label{remark_bounded_variation_inequality}
\end{align}
This domination also  holds for  $A^*=0$ in the same way by its definition as a limit in \eqref{def_w2} and because \eqref{signofderiv} holds for every $A > 0$.
Finally for  $x\in(-\infty,A^*)$, the equality $u_{A^*}(x) = g(x)$ holds  by definition. As a result, we conclude that
\begin{align}
u_{A^*}(x) \geq g(x) \quad \text{ for } x \in \R. \label{lemma_u_greater_than_g}
\end{align}

For the supermartingale property, we shall use the following result (see Appendix for a proof):
\begin{lemma} \label{lemma_generator_zero_seller}
We have $(\mathcal{L}-r) u_{A^*}(x) \leq 0$ for every $x \in (0,\infty) \backslash \{ A^* \}$.
\end{lemma}

Finally,   inequality \eqref{lemma_u_greater_than_g} and Lemma \ref{lemma_generator_zero_seller}  yield the optimality of  $u_{A^*}$ similarly to Theorem \ref{optimality_buyer}. Hence, we have
\begin{theorem} \label{proposition_optimality_case_1} The candidate function $u_{A^*} (x)$, with $A^*$ given by \eqref{smooth_fit_seller},  is  optimal for \eqref{seller_problem}. That is,
\begin{align*}
u_{A^*}(x) = u(x) = \sup_{\tau \in \S} \E^x \left[ e^{-r \tau} g(X_\tau) 1_{\left\{ \tau < \infty \right\}} \right].
\end{align*}
\end{theorem}
With this result, we have solved for the putable step-down and callable step-up default swaps in view of the decompositions by Propositions \ref{prop-V} and \ref{prop-U}.

\section{Numerical Examples} \label{section_numerical}
In this section, we numerically  illustrate the investor's optimal exercise strategy and the credit spread behaviors, where the underlying spectrally negative \lev process is assumed to have \emph{hyperexponential jumps}   of the form
\begin{equation}
  X_t  - X_0=\mu t+\sigma B_t - \sum_{n=1}^{N_t} Z_n, \quad 0\le t <\infty. \label{levy_canonical}
\end{equation}
Here $B=\{B_t; t\ge 0\}$ is a standard Brownian motion, $N=\{N_t; t\ge 0\}$ is a Poisson process with arrival rate $\lambda$, and  $Z = \left\{ Z_n; n = 1,2,\ldots \right\}$ are i.i.d.  hyperexponential  random variables with density function
\begin{align*}
f (z)  = \sum_{i=1}^m \alpha_i \eta_i e^{- \eta_i z}, \quad z > 0,
\end{align*}
for some $0 < \eta_1 < \cdots < \eta_m < \infty$.

As discussed in \cite{Egami_Yamazaki_2010_2}, the scale function of the \lev process of the form \eqref{levy_canonical} admits  analytic form and can approximate that of any spectrally negative \lev process with a completely monotone \lev measure.   For our numerical examples, we consider the process in the form (\ref{levy_canonical}) with $Z$ replaced by Pareto random variables with distribution function $F(t) = 1-(1+5t)^{-1.2}$ for $t \geq 0$.   We use the approximation to its scale function computed in \cite{Egami_Yamazaki_2010_2} where they adopted the hyperexponential fitting algorithm given by \cite{Feldmann_1998}.
We refer to \cite{Egami_Yamazaki_2010_2, Feldmann_1998} for the detailed fitting procedure and  fitted parameters.

For both callable and putable default swaps, we consider the following cases:
\begin{itemize}
\item Step-Down default swap with $\hat{p}/p = \hat{\alpha} / \alpha=0.5$,
\item Step-Up default swap with $\hat{p}/p = \hat{\alpha} / \alpha = 1.5$.
\end{itemize}
Hence, there are in total 4 cases.  The model parameters are $r =0.03$, $\sigma = 0.2$,  $\alpha=1$, $x=1.5$ and $\gamma = 50$bps, unless specified otherwise. We shall adjust the values of $\lambda$ and $\mu$ so that the risk-neutral condition $\psi(1) = r$ holds.

By    symmetry  (see Section \ref{subsection_symmetry}), the optimal stopping problems   for   callable step-down and putable step-up default swaps are equivalent, while callable step-up and putable step-down default swaps are  equivalent.  Figure \ref{figure_threshold} shows the optimal thresholds $B^*$ for  callable step-down/putable step-up default swaps  and $A^*$ for callable step-up/putable step-down   default swaps.      Both $B^*$ and $A^*$ are decreasing in the premium $p$, as is  intuitive. Also,   $A^*$ and $B^*$ rise as default risk $\lambda$ increases.

\begin{figure}[htb]
\begin{center}
\includegraphics[scale=0.55]{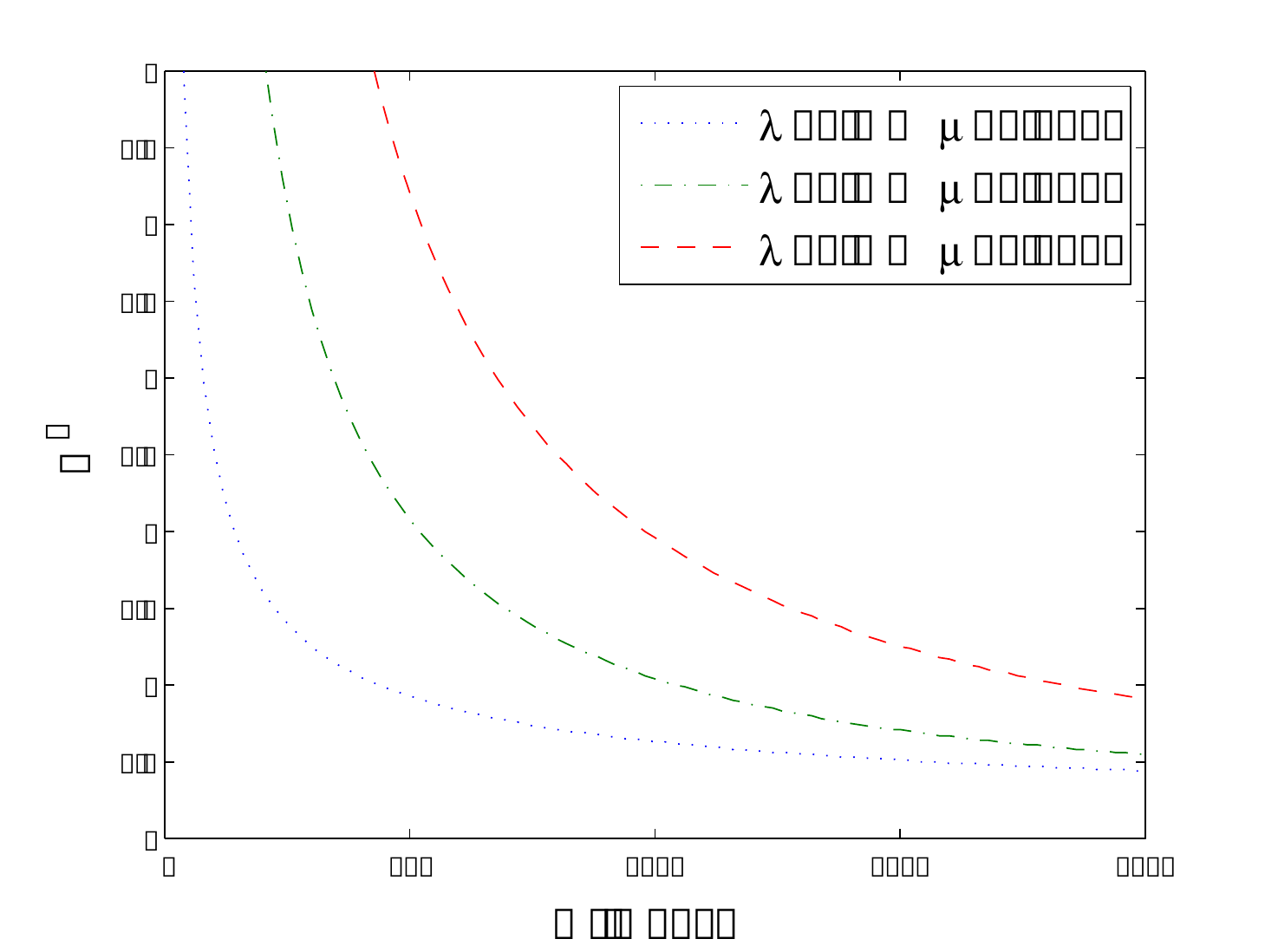} \includegraphics[scale=0.55]{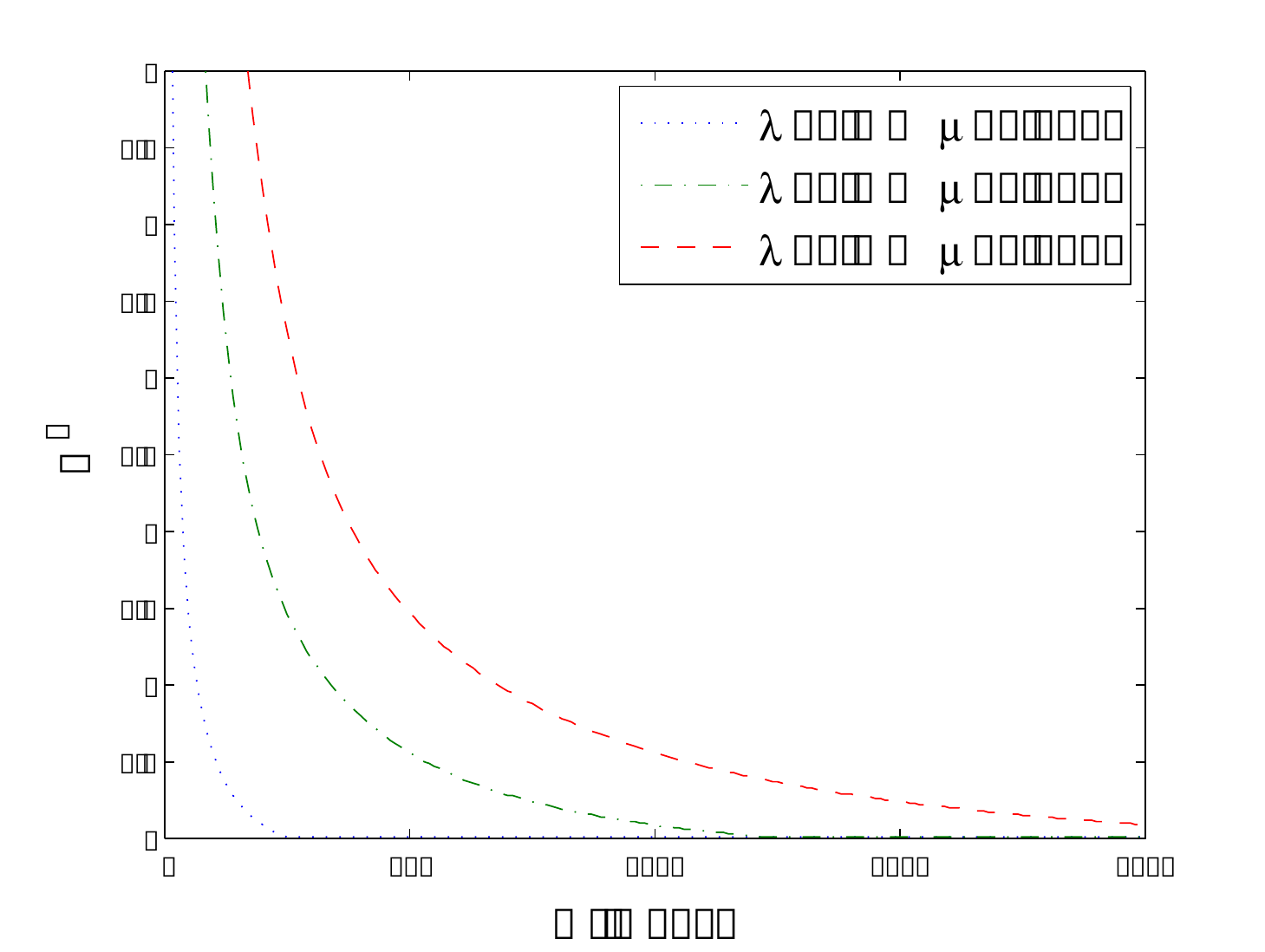}
\caption{\small{(Left): optimal thresholds for the callable step-down/putable step-up default swap. (Right): optimal thresholds for the callable step-up/putable step-down default swap.}} \label{figure_threshold}
\end{center}
\end{figure}

Figure \ref{premium_callable}   shows the credit spread $p^*$ as a function of the distance-to-default $x$ for the callable and putable  step-down default swaps, with  the vanilla CDS  as benchmark (see \eqref{def_p_pertpetual}). We first compute the contract values, which are  monotone in $p$, and then determine $p^*$ by a bisection method. As $x$ increases, meaning lower default risk, the credit spread  $p^*$ reduces.  The callable step-down default swap spreads are naturally higher than the vanilla case due to  the embedded step-down option.  In contrast,   the  putable step-down   spreads are lower than the vanilla case because the buyer is subject to the step-down exercise by the seller.

\begin{figure}[htbp]
\begin{center}
  \includegraphics[scale=0.55]{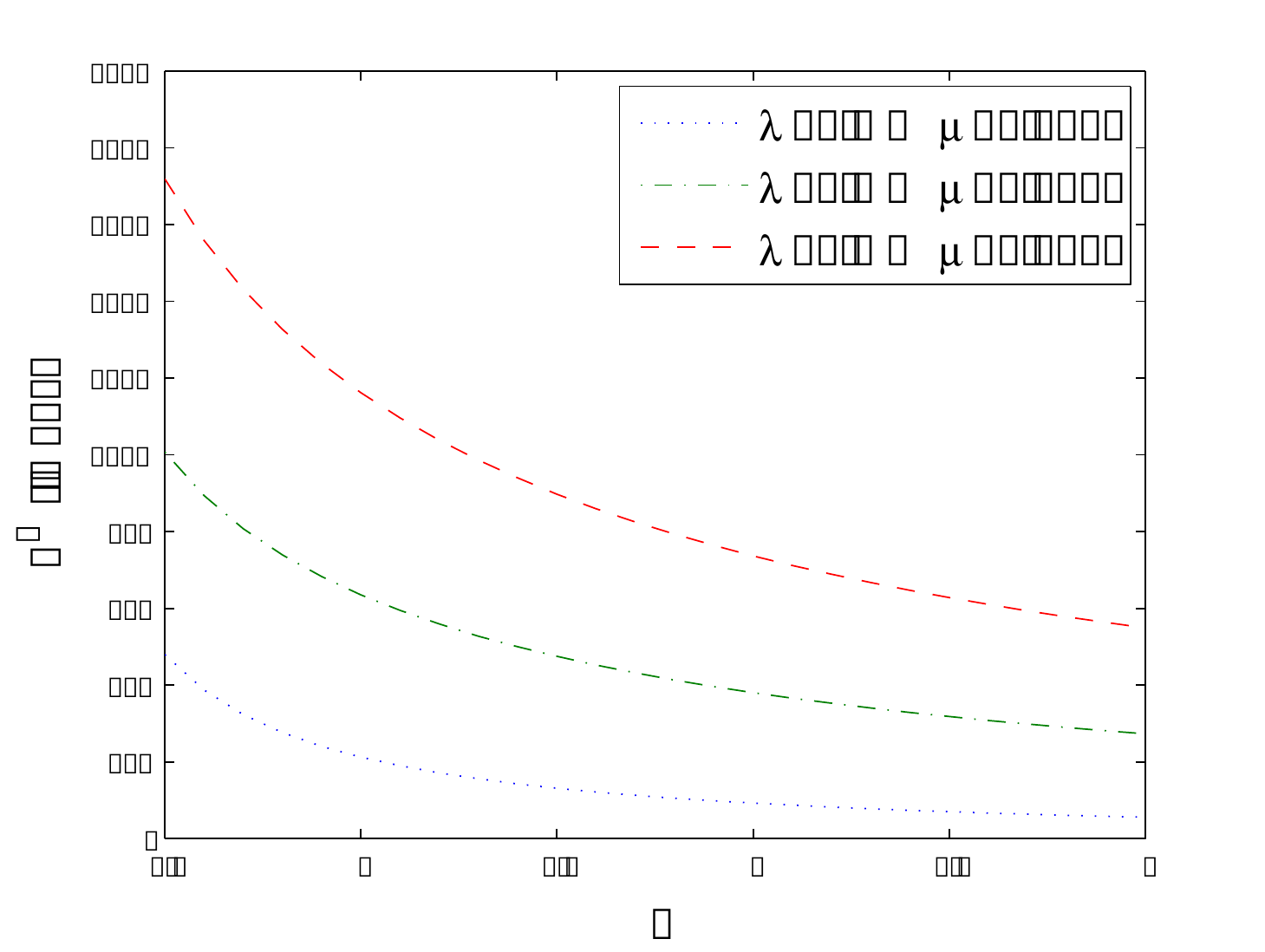}\includegraphics[scale=0.55]{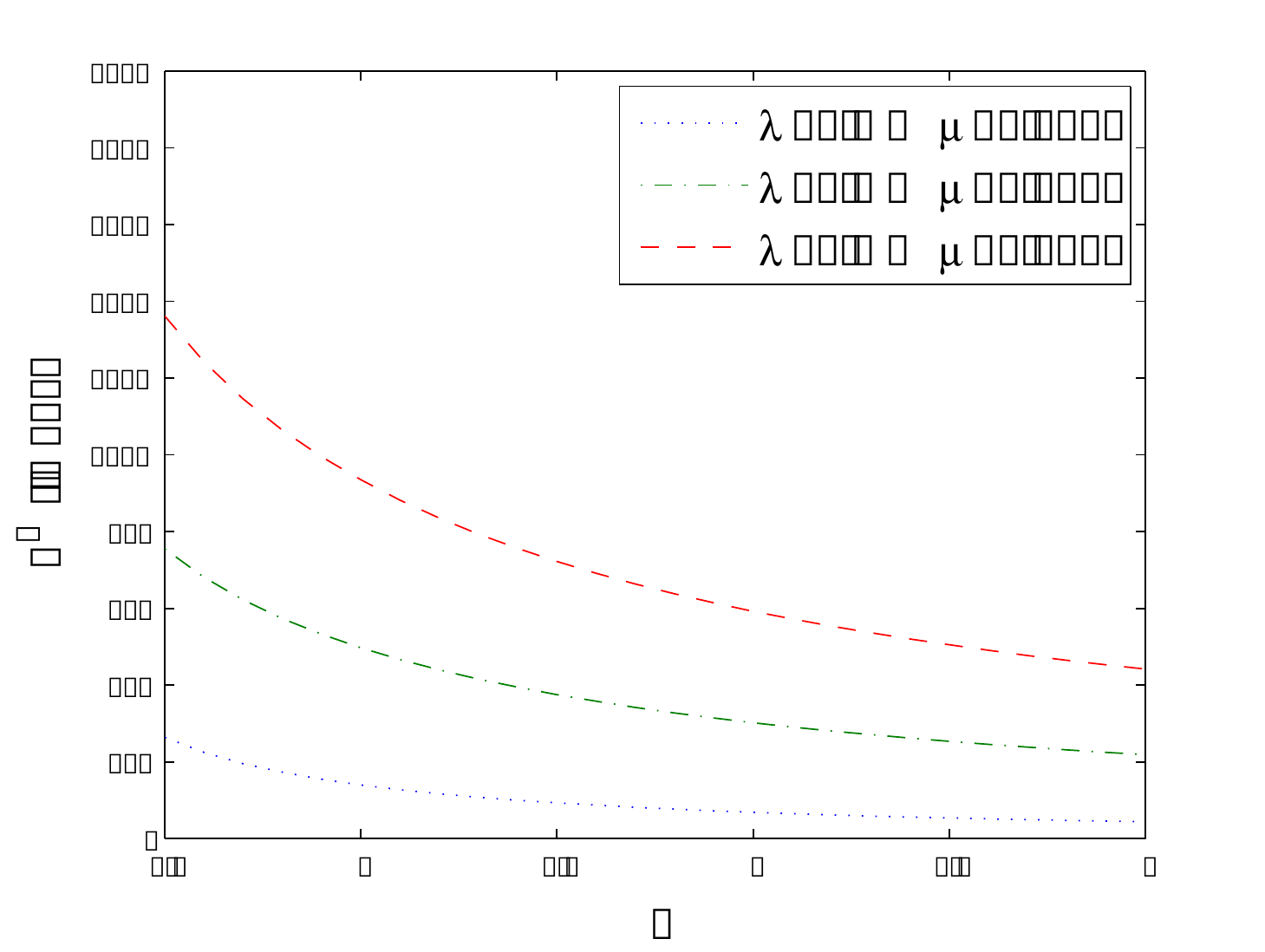}\\
\includegraphics[scale=0.55]{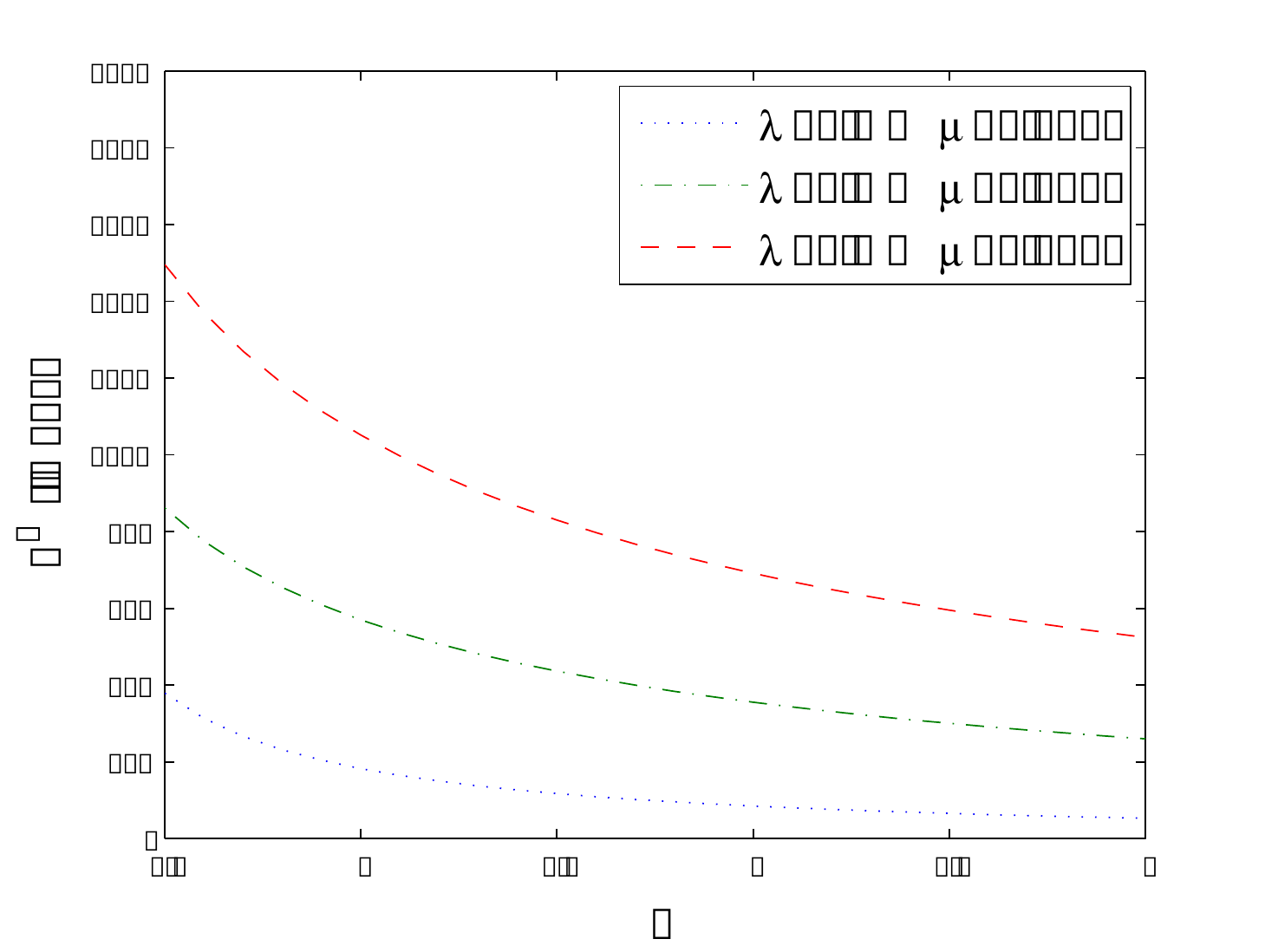}
\caption{\small{Credit spreads vs distance-to-default for the callable step-down (top-left), putable step-down (top-right) and vanilla (bottom) default swaps.}} \label{premium_callable}
\end{center}
\end{figure}

\section{The Finite-maturity case} \label{section_finite_maturity}

We now consider the finite-maturity case, and study how the solutions can be approximated using the results obtained in the previous sections.   We first formulate the finite-maturity American callable/putable step-up/down default swaps by modifying the results in Sections \ref{sect-callableCDS} and \ref{sect-putableCDS}.  We then show that their value functions, optimal strategies, and credit spreads can be efficiently approximated by our analytical results on the perpetual case.

\subsection{Finite-Maturity Formulation}
Let $T\in (0,\infty)$ be a given finite maturity and define
\begin{align}
\S_T : = \left\{ \mathbb{F}\textrm{-stopping time } \tau: \tau \leq \theta \wedge T \; \text{ a.s. }\right\} \label{def_S_T}
\end{align}
be the set of all stopping times smaller than or equal to $\theta \wedge T$.  The buyer's and seller's maximal expected cash flows, respectively,  are given by,
\begin{align*}\bar{V}&(x, T; p,\ph, \alpha, \ah, \gamma)  \\
&:=\sup_{\tau \in \S_T} \E^x \left[ -\int_0^\tau e^{-rt} p\,\diff t-\int_\tau^{\theta \wedge T} e^{-rt}\ph\,\diff t -e^{-r \tau} \gamma  1_{\{\tau <\theta \wedge T\}} + e^{-r \theta} 1_{\{\theta < T\}} (\ah 1_{\{\tau <\theta\}} +\alpha 1_{\{\tau =\theta \}}) \right], \\
\bar{U}&(x, T; p,\ph, \alpha, \ah, \gamma)  \\
&:=\sup_{\tau \in \S_T} \E^x \left[ \int_0^\tau e^{-rt} p\,\diff t + \int_\tau^{\theta \wedge T} e^{-rt}\ph\,\diff t -e^{-r \tau} \gamma  1_{\{\tau <\theta \wedge T\}} - e^{-r \theta} 1_{\{\theta < T\}} (\ah 1_{\{\tau <\theta\}} +\alpha 1_{\{\tau =\theta \}}) \right].
\end{align*}
Here, the contract is terminated at default $\theta$ or maturity $T$, whichever comes first. For the callable case, the buyer can exercise anytime (strictly) before the contract termination.   Since $\PP\{\theta = T\}=0$ for any $x > 0$,  we may interpret $\{ \tau = T\}$ as the event that the option expires without being exercised. When default happens before maturity (on $\{ \theta < T\}$), the default payment is $\ah$ if the buyer has exercised ($\tau < \theta$) and is $\alpha$ otherwise ($\tau = \theta$).  The interpretation for the putable case is similar.

\subsection{Symmetry and Decomposition}
As we shall describe below, the symmetry and decomposition we attained in Section \ref{sect-overview} for the perpetual case can be extended to the finite-maturity case.

To see this, we can decompose the buyer's maximal expected cash flow as
\begin{align} \label{V_x_T_decomposition}
\begin{split}
\bar{V}(x, T)
&= \sup_{\tau \in \S_T} \E^x \left[\int_\tau^{\theta \wedge T} e^{-rt} \pcheck\, \diff t - \int_0^{\theta \wedge T} e^{-rt} p\, \diff t - e^{-r \tau} \gamma 1_{\{\tau < \theta \wedge T\}} + e^{-r \theta} 1_{\{\theta < T\}} (-  \acheck 1_{\{\tau < \theta \}}  +\alpha ) \right] \\
  &= \underbrace{\sup_{\tau \in \S_T} \E^x \left[ \int_\tau^{\theta \wedge T} e^{-rt}\pcheck \,\diff t - e^{-r \tau} \gamma  1_{\{\tau <\theta \wedge T \}} -  e^{-r \theta} \acheck  1_{\{\tau <\theta < T\}}  \right]}_{=:f(x,T)} + \bar{C}(x, T;p,\alpha),
\end{split}
\end{align}
where $\pcheck$ and $\acheck$ are the same as in \eqref{acheck_pcheck} and $\bar{C}$ is the value of a standard finite-maturity CDS defined as in \eqref{def_C_bar}. By   $\{ \tau  < \theta \wedge T\} =\{ X_\tau >  0, \tau < T \}$ for every $\tau \in \mathcal{S}_T$,
 \begin{align*}
{f(x,T)}= \sup_{\tau \in \S_T}\E^x \left[ e^{-r \tau} h(X_\tau, T - \tau)\right],
\end{align*}
with
\begin{align*}h(x, t) := 1_{\{x > 0, t > 0 \}}\left( \E^x \left[ \int_0^{\theta \wedge t} e^{-rs} \pcheck \, \diff s\,  - e^{-r \theta} \acheck 1_{\{ \theta < t \}}\, \right] - \gamma \right)= 1_{\{x > 0, t > 0 \}} \big(\bar{C}(x, t;-\pcheck,-\acheck)- \gamma \big).\end{align*}
The case for the seller is similar.
Consequently, Propositions \ref{prop-V} and \ref{prop-U} can be extended to the finite-maturity case.

\begin{proposition}For all $x, T > 0$, we have the decomposition
\begin{align} \label{value_function_decomposition_finite}
\begin{split}
\bar{V}(x, T; p,\ph, \alpha, \ah, \gamma) &= \bar{C}(x, T;p,\alpha)+\bar{v}(x, T;-\pcheck, -\acheck,\gamma), \\
\bar{U}(x, T; p,\ph, \alpha, \ah, \gamma) &= - \bar{C}(x, T;p,\alpha)+\bar{u}(x, T;-\pcheck,- \acheck,\gamma),
\end{split} \end{align}
where
 \begin{align*}\bar{v}(x, T;\kappa,a, K) &:=\sup_{\tau \in \S_T}\E^x \left[ e^{-r \tau}  \left(\bar{C}(X_{\tau}, T - \tau; \kappa,a) -K\right)^+  1_{\{ \tau < \theta \wedge T\}}  \right], \quad \text{ and }\\
\bar{u}(x, T;\kappa,a,K) &:=\sup_{\tau \in \S_T}\E^x \left[ e^{-r \tau}  \left(-\bar{C}(X_{\tau}, T-\tau; \kappa,a)-K\right)^+  1_{\{ \tau < \theta \wedge T \}}  \right].
\end{align*}
\end{proposition}

Clearly,
\begin{align}
\bar{v}(x, T;\kappa,a, K) = \bar{u}(x, T;-\kappa,-a, K), \label{identity_finite_maturity}
\end{align}
and Table \ref{table:4kinds} holds for the finite-maturity case as well.  Moreover,
``put-call parity" and symmetry identities also hold simply by  replacing $C(x)$ with $\bar{C}(x,T)$. To summarize,  \begin{align*}\bar{V}(x, T;p, \ph, \alpha, \ah, \gamma ) - \bar{U}(x, T;p, 2p-\ph, \alpha, 2\alpha-\ah, \gamma ) &= 2\, \bar{C}(x, T;p, \alpha),\\
\bar{V}(x, T;p, \ph, \alpha, \ah, \gamma ) + \bar{U}(x, T;p, 2p-\ph, \alpha, 2 \alpha-\ah, \gamma )&= 2\,\bar{v}(x, T;\ph-p, \ah-\alpha, \gamma ).
\end{align*}

  While these identifications are analogous to the perpetual case,  the computation of the value function for the finite-maturity case \eqref{value_function_decomposition_finite} is significantly more difficult.  Whereas  $\bar{C}$ can be computed using standard techniques such as Laplace inversion and simulation, the computation of $\bar{v}$ and $\bar{u}$  must involve a free-boundary problem of PIDE.  Moreover, in light of the non-standard nature of our problem such as the discontinuity of the payoff function and early termination due to default, it is not clear if any standard numerical method can achieve reasonable accuracy.  It is also noted that one needs to focus on a certain type of \lev process (and its infinitesimal generator), and the results are significantly limited compared to our results on the perpetual case, which are applicable to a general spectrally negative \lev process.  For this reason, we take an analytical approach by utilizing greatly our analytical solutions for the perpetual case.

\subsection{Analytical Bounds and Asymptotic Optimality} In view of the computational challenges involving $\bar{v}(x,T)$ and $\bar{u}(x,T)$, here we discuss  how these can be approximated using the analytical value functions from the perpetual case, namely,  $v(x)$ and $u(x)$ (see \eqref{U-smaller_than_b_star}, \eqref{u_tilde_a} and \eqref{def_w2}).

As in the perpetual case, in order to compute the value function for all four cases (callable/putable and step-up/down), it suffices to obtain  $\bar{v}(x,T) := \bar{v}(x, T; -\pcheck, -\acheck, \gamma)$ and $\bar{u}(x, T) := \bar{u}(x, T; -\pcheck, -\acheck, \gamma)$ for $\pcheck, \acheck > 0$ thanks to \eqref{identity_finite_maturity}.   As in \eqref{V_x_T_decomposition}, we have the identities:
\begin{align}
\bar{v}(x, T) &= \sup_{\tau \in \S_T} \E^x \left[ \int_\tau^{\theta \wedge T} e^{-rt}\pcheck \,\diff t - e^{-r \tau} \gamma  1_{\{\tau <\theta \wedge T \}} -  e^{-r \theta} \acheck  1_{\{\tau <\theta < T\}}  \right],  \label{v_bar_identity} \\
\bar{u}(x, T) &= \sup_{\tau \in \S_T}  \E^x \left[ -\int_\tau^{\theta \wedge T} e^{-rt}\pcheck \,\diff t - e^{-r \tau} \gamma  1_{\{\tau <\theta \wedge T \}} +  e^{-r \theta} \acheck  1_{\{\tau <\theta < T\}}  \right]. \label{u_bar_identity}
\end{align}

We shall first show that the function
\begin{align*}
\tilde{v}(x, T) := v(x) - \E^x \left[ 1_{\{\theta \geq T \}}\int_T^\theta e^{-rt}\pcheck \,\diff t \right]
\end{align*}
can approximate $\bar{v}(x, T)$, for $x, T > 0$, with some suitable analytical bounds.
\begin{lemma} \label{lemma_bound_v_tilde}
For $\pcheck, \acheck > 0$ and $\gamma \geq 0$, we have
\begin{align*}
\tilde{v}(x,T) \leq  \bar{v}(x,T)  \leq \tilde{v}(x,T) +\E^x \left[ 1_{\{\theta \geq T\}} \left( e^{-r T} \gamma    + e^{-r \theta} \acheck  \right) \right],
\end{align*}
for  $x , T > 0$.
\end{lemma}
\begin{proof}  Using  \eqref{v_bar_identity} and  that $\tau \leq T$ a.s.\ for any $\tau \in \S_T$ by \eqref{def_S_T}, we write
\begin{align*}
\bar{v}(x,T) &= \sup_{\tau \in \S_T} \E^x \left[ \int_\tau^{\theta} e^{-rt}\pcheck \,\diff t - 1_{\{ \theta \geq T \}} \int_T^\theta e^{-rt}\pcheck \,\diff t  - e^{-r \tau} \gamma  1_{\{\tau <\theta \wedge T \}} -  e^{-r \theta} \acheck  1_{\{\tau <\theta < T\}}  \right].
\end{align*}
Observing that, for any $\tau \in \S_T$,
\begin{align}
e^{-r \tau}1_{\{ \tau < \theta \wedge T\}} = e^{-r \tau} 1_{\{ \tau < \theta\}} - e^{-r \tau} 1_{\{ T = \tau < \theta \}} =  e^{-r \tau} 1_{\{ \tau < \theta\}} - e^{-r T} 1_{\{ T = \tau < \theta \}} \geq  e^{-r \tau} 1_{\{ \tau < \theta\}} -  e^{-r T}1_{\{ \theta \geq T\}}, \label{inequality_gamma_term}
\end{align}
and $  1_{\{\tau <\theta < T\}} \geq   1_{\{\tau <\theta \}} -   1_{\{\theta \geq T\}}$, we obtain the inequality
\begin{align}
\begin{split} \bar{v}(x,T)
&\leq \sup_{\tau \in \S_T} \E^x \left[ \int_\tau^{\theta} e^{-rt}\pcheck \,\diff t - e^{-r \tau} \gamma  1_{\{\tau <\theta\}}  -  e^{-r \theta} \acheck  1_{\{\tau <\theta \}}  \right]  \\ &+\E^x \left[ 1_{\{ \theta \geq T \}} \left( - \int_T^\theta e^{-rt}\pcheck \,\diff t   + e^{-r T} \gamma    + e^{-r \theta} \acheck  \right) \right].
\end{split} \label{v_bar_X_t_bound}
\end{align}
Recall from \eqref{V_x_decomposition} that
\begin{align}
v(x) &= \sup_{\tau \in \S} \E^x \left[ 1_{\{ \tau < \infty \}} \left( \int_\tau^{\theta} e^{-rt}\pcheck \,\diff t - e^{-r \tau} \gamma  1_{\{\tau <\theta\}}  -  e^{-r \theta} \acheck  1_{\{\tau <\theta \}} \right)    \right]  \label{v_x_second_def}
\end{align}
and  $\S_T \subset \S$, $v(x)$ dominates the first expectation of the right hand side in \eqref{v_bar_X_t_bound}.  Hence we have the desired upper bound.

For the lower bound, by \eqref{v_x_second_def},
\begin{align*}
v(x)
&\leq \sup_{\tau \in \S} \E^x \left[ 1_{\{ \tau < \infty \}}\int_\tau^{\theta} e^{-rt}\pcheck \,\diff t - e^{-r \tau} \gamma  1_{\{\tau <\theta \wedge T\}}  -  e^{-r \theta} \acheck  1_{\{\tau <\theta < T\}}  \right].
\end{align*}
In view of the integrand on the right hand side, in the event the contract has not been terminated until $T$, it is optimal to exercise at $T$ because waiting further would simply reduce the cash flow at rate $\pcheck$.  Therefore, the optimal stopping time must be in  $\S_T$ and
\begin{align*}
v(x) &\leq \sup_{\tau \in \S_T} \E^x \left[ \int_\tau^{\theta} e^{-rt}\pcheck \,\diff t - e^{-r \tau} \gamma  1_{\{\tau <\theta \wedge T\}}  -  e^{-r \theta} \acheck  1_{\{\tau <\theta < T\}}  \right] \\
&= \sup_{\tau \in \S_T} \E^x \left[ \int_\tau^{\theta \wedge T} e^{-rt}\pcheck \,\diff t - e^{-r \tau} \gamma  1_{\{\tau <\theta \wedge T\}}  -  e^{-r \theta} \acheck  1_{\{\tau <\theta < T\}}  \right] +  \E^x \left[ \int_{\theta \wedge T}^\theta e^{-rt}\pcheck \,\diff t \right] \\
&= \bar{v}(x,T) +  \E^x \left[ \int_{\theta \wedge T}^\theta e^{-rt}\pcheck \,\diff t \right],
\end{align*}
which gives the lower bound.
\end{proof}

Similarly, to approximate $\bar{u}(x, T)$,  we define
\begin{align*}
\tilde{u}(x, T) := u(x) + \E^x \left[ 1_{\{ \theta \geq T\}}\int_{T}^\theta e^{-rt}\pcheck \,\diff t \right].
\end{align*}

\begin{lemma} \label{lemma_bound_u_tilde}
For $\pcheck, \acheck > 0$ and $\gamma \geq 0$,
\begin{align*}
\tilde{u}(x, T) - \E^x \left[ 1_{\{ \theta \geq T\}} \left( \int_{T}^\theta e^{-rt}\pcheck \,\diff t  + e^{-r \theta} \acheck \right) \right] \leq \bar{u}(x,T)   \leq \tilde{u}(x, T)  + \E^x \left[  1_{\{\theta \geq T \}} e^{-r T} \gamma  \right],
\end{align*}
for every $x > 0$ and $T > 0$.
\end{lemma}
\begin{proof}
By \eqref{u_bar_identity} and   $\tau \leq T$ a.s.\ for any $\tau \in \S_T$ by \eqref{def_S_T},
\begin{align*}
\bar{u}(x,T)
&= \sup_{\tau \in \S_T} \E^x \left[ -\int_\tau^{\theta} e^{-rt}\pcheck \,\diff t + 1_{\{ \theta \geq T \}}\int_T^\theta e^{-rt}\pcheck \,\diff t  - e^{-r \tau} \gamma  1_{\{\tau <\theta \wedge T \}} +  e^{-r \theta} \acheck  1_{\{\tau <\theta < T\}}  \right].
\end{align*}
By \eqref{inequality_gamma_term} and that $\S_T \subset \S$,
\begin{align*}
\bar{u}(x,T)
&\leq \sup_{\tau \in \S_T} \E^x \left[ -\int_\tau^{\theta} e^{-rt}\pcheck \,\diff t- e^{-r \tau} \gamma  1_{\{\tau <\theta\}} +  e^{-r \theta} \acheck  1_{\{\tau <\theta \}}  \right]   + \E^x \left[  1_{\{\theta \geq T \}}  \left( \int_{T}^\theta e^{-rt}\pcheck \,\diff t    + e^{-r T} \gamma \right)\right] \\
&\leq \sup_{\tau \in \S} \E^x \left[ 1_{\{ \tau < \infty \}} \left( -\int_\tau^{\theta} e^{-rt}\pcheck \,\diff t- e^{-r \tau} \gamma  1_{\{\tau <\theta\}} +  e^{-r \theta} \acheck  1_{\{\tau <\theta \}}  \right) \right]   + \E^x \left[  1_{\{\theta \geq T \}}  \left( \int_{T}^\theta e^{-rt}\pcheck \,\diff t    + e^{-r T} \gamma \right)\right] ,
\end{align*}
which is the desired upper bound.

On the other hand, because $\{ \tau < \theta \} \subset \{ \tau < \theta < T \} \cup \{ \theta \geq T \}$,
\begin{align}
u(x) &= \sup_{\tau \in \S} \E^x \left[ 1_{\{ \tau < \infty \}} \left( -\int_\tau^{\theta} e^{-rt}\pcheck \,\diff t- e^{-r \tau} \gamma  1_{\{\tau <\theta\}} +  e^{-r \theta} \acheck  1_{\{\tau <\theta \}} \right)  \right]  \nonumber \\
&\leq \sup_{\tau \in \S} \E^x \left[ -\int_\tau^{\theta \wedge T} e^{-rt}\pcheck \,\diff t- e^{-r \tau} \gamma  1_{\{\tau <\theta \wedge T\}} +  e^{-r \theta} \acheck  1_{\{\tau <\theta < T \}} \right]    +  \E^x \left[ e^{-r \theta} \acheck  1_{\{\theta \geq T \}} \right] . \label{gap_u_bound}
\end{align}
Because in the integrand on the right-hand side, the payoff after $T$ is uniformly zero, we can replace $\S$ with $\S_T$ and this gives the lower bound.\end{proof}

Define the error functions
\begin{align*}
\epsilon_v (x,T) &:= \bar{v} (x,T) - \tilde{v}(x,T), \\
\epsilon_u (x,T) &:= \bar{u} (x,T) - \tilde{u}(x,T),
\end{align*}
for every $x > 0$ and $T > 0$.
By Lemmas \ref{lemma_bound_v_tilde} and \ref{lemma_bound_u_tilde}, these have bounds:
\begin{align}
0 \leq  &\epsilon_v (x,T)  \leq \E^x \left[ 1_{\{\theta \geq T \}}  \left( e^{-r T} \gamma    + e^{-r \theta} \acheck   \right) \right],  \label{bound_epsilon_v}\\
- \E^x \left[ 1_{\{ \theta \geq T\}} \left( \int_{T}^\theta e^{-rt}\pcheck \,\diff t  + e^{-r \theta} \acheck \right) \right] \leq &\epsilon_u (x,T) \leq \E^x \left[  1_{\{\theta \geq T \}} e^{-r T} \gamma  \right],  \label{bound_epsilon_u}
\end{align}
for every $x > 0$ and $T > 0$.
If the exercise fee $\gamma$ is set zero or sufficiently small (recall that the exercise fee $\gamma$ is supposed to be much smaller than the change in default payment $\acheck$), the bound \eqref{bound_epsilon_v} is expected to be small.  Indeed, in light of the expectation $\E^x \left[ 1_{\{\theta \geq T \}} e^{-r \theta}  \right]$, the discount $e^{-r \theta}$ and  the indicator function $1_{\{ \theta \geq T\}}$ keep its value small when $\theta$ is large and when it is small, respectively.  The bound \eqref{bound_epsilon_u} is in comparison less tight because of the expectation $\E^x [ 1_{\{ \theta \geq T\}} \int_{T}^\theta e^{-rt}\pcheck \,\diff t  ]$ on the lower bound.  In the analytical proof of Lemma \ref{lemma_bound_u_tilde} (in particular \eqref{gap_u_bound}), this term cannot be removed, but $\epsilon_u (x,T)$ is expected to be closer to its upper bound than to its lower bound in \eqref{bound_epsilon_u}.

We shall show that the error functions vanish in the limit as $T \uparrow \infty$ which also implies that $\bar{v}(x,T)$ converges to the perpetual value function $v(x)$.  We further show upon some suitable conditions that these error functions also converge to zero as $x \uparrow \infty$ and $x \downarrow 0$.

The first convergence result as $T \uparrow \infty$ is immediate because
\begin{align*}
\lim_{T \uparrow \infty}\E^x \left[ 1_{\{T \leq \theta\}}\int_T^\theta e^{-rt}\pcheck \,\diff t \right] =\lim_{T \uparrow \infty} \E^x \left[ e^{-r T} \gamma  1_{\{\theta \geq T \}}  \right]=
\lim_{T \uparrow \infty}\left[ e^{-r \theta} \acheck  1_{\{\theta \geq T\}} \right] = 0.
\end{align*}
\begin{proposition} \label{convergence_u_v_bar}
For every fixed $x > 0$, as $T \uparrow \infty$,
\begin{enumerate}
\item $\epsilon_v (x, T) \rightarrow 0$ and $\epsilon_u (x, T) \rightarrow 0$,
\item $\bar{v} (x, T) \rightarrow v(x)$ and $\tilde{v} (x, T) \rightarrow v(x)$,
\item $\bar{u} (x, T) \rightarrow u(x)$ and $\tilde{u} (x, T) \rightarrow u(x)$.
\end{enumerate}
\end{proposition}

This proposition shows, for each of the callable/putable step-up/down cases,  the asymptotic optimality  as $T \uparrow \infty$ of the value functions $V(x)$ and $U(x)$ of the perpetual case; namely,
\begin{align}
V(x,T) &\rightarrow V(x)  \label{asymp_optimality_buyer}, \\
U(x,T) &\rightarrow U(x), \label{asymp_optimality_seller}
\end{align}
as $T \uparrow \infty$ for every fixed $x > 0$.  This also implies the effectiveness of the approximation using the strategies $\tau_{B^*}^+ \wedge T \in \S_T$ and $\tau_{A^*}^- \wedge T \in \S_T$.   To this end, let $\bar{v}(x, T, \tau_{B^*}^+ \wedge T)$,  $\bar{u}(x, T, \tau_{A^*}^- \wedge T)$, $\bar{V}(x, T, \tau_{B^*}^+ \wedge T)$ and  $\bar{U}(x, T, \tau_{A^*}^- \wedge T)$ be the expected values corresponding to the strategy $\tau_{B^*}^+ \wedge T $ and $\tau_{A^*}^- \wedge T $.
Observing that  $\{\tau_{B^*}^+ \wedge T <\theta \wedge T \} = \{\tau_{B^*}^+ <\theta \wedge T \} \cup \{T <\theta \wedge T \} = \{\tau_{B^*}^+ <\theta \wedge T \}$ and $\{\tau_{B^*}^+ \wedge T <\theta < T\} = \{\tau_{B^*}^+  <\theta < T\}$ (since $\tau_{B^*}^+ \leq \theta$ implies $\tau_{B^*}^+ < T$ on $\{ \theta < T\}$), we can write
\begin{align} \label{v_bar_tau_B}
\begin{split}
\bar{v}(x, T, \tau_{B^*}^+ \wedge T) &= \E^x \left[ \int_{\tau_{B^*}^+ \wedge T}^{\theta \wedge T} e^{-rt}\pcheck \,\diff t - e^{-r (\tau_{B^*}^+ \wedge T)} \gamma  1_{\{\tau_{B^*}^+ \wedge T <\theta \wedge T \}} -  e^{-r \theta} \acheck  1_{\{\tau_{B^*}^+ \wedge T <\theta < T\}}  \right] \\
 &= \E^x \left[ \int_{\tau_{B^*}^+ \wedge T}^{\theta \wedge T} e^{-rt}\pcheck \,\diff t - e^{-r (\tau_{B^*}^+ \wedge T)} \gamma  1_{\{\tau_{B^*}^+  <\theta \wedge T \}} -  e^{-r \theta} \acheck  1_{\{\tau_{B^*}^+ <\theta < T\}}  \right].
\end{split}
\end{align}
Similarly, we conclude that
\begin{align*}
\bar{u}(x, T, \tau_{A^*}^- \wedge T)
 &= \E^x \left[ -\int_{\tau_{A^*}^- \wedge T}^{\theta \wedge T} e^{-rt}\pcheck \,\diff t - e^{-r (\tau_{A^*}^- \wedge T)} \gamma  1_{\{\tau_{A^*}^-  <\theta \wedge T \}} + e^{-r \theta} \acheck  1_{\{\tau_{A^*}^- <\theta < T\}}  \right].
\end{align*}
By \eqref{value_function_decomposition_finite} and \eqref{identity_finite_maturity}, $\bar{V}(x,T,\tau_{B^*}^+ \wedge T)$ and ${\bar{U}(x,T,\tau_{A^*}^- \wedge T)}$ can be obtained by adding/subtracting $\bar{C}(x,T)$.
\begin{proposition}
The asymptotic optimality of $\tau_{B^*}^+ \wedge T$ and $\tau_{A^*}^- \wedge T$ holds. That is, for every $x > 0$,
\begin{enumerate}
\item $ {\bar{v}(x,T,\tau_{B^*}^+ \wedge T)}/ {\bar{v}(x,T)} \rightarrow 1$,
\item $ {\bar{u}(x,T,\tau_{A^*}^- \wedge T)} /{\bar{u}(x,T)} \rightarrow 1$,
\item ${\bar{V}(x,T,\tau_{B^*}^+ \wedge T)} / {\bar{V}(x,T)} \rightarrow 1$,
\item  ${\bar{U}(x,T,\tau_{A^*}^- \wedge T)} /{\bar{U}(x,T)} \rightarrow 1$,
\end{enumerate}
as $T \uparrow \infty$.
\end{proposition}
\begin{proof}
By the definition of $\bar{v}(x,T)$ in \eqref{v_bar_identity} as a supremum and because $\tau_{B^*}^+ \wedge T \in \S_T$, we have $\bar{v}(x,T,\tau_{B^*}^+ \wedge T) \leq \bar{v}(x,T)$ for any $T > 0$.  By this and Proposition \ref{convergence_u_v_bar},
\begin{align*}
\limsup_{T \uparrow \infty}\bar{v}(x,T,\tau_{B^*}^+ \wedge T) \leq \lim_{T \uparrow \infty}\bar{v}(x,T) = v(x).
\end{align*}
On the other hand, applying Fatou's lemma in \eqref{v_bar_tau_B} and because $\tau_{B^*}^+ \wedge T \rightarrow \tau_{B^*}^+$  on $\{\tau_{B^*}^+ < \infty \}$ and $\theta \wedge T \rightarrow \theta$  on $\{\theta < \infty \}$  as $T \uparrow \infty$,
\begin{align*}
\liminf_{T \uparrow \infty}\bar{v}(x, T, \tau_{B^*}^+ \wedge T) &\geq \E^x \left[ \liminf_{T \uparrow \infty} \left( \int_{\tau_{B^*}^+ \wedge T}^{\theta \wedge T} e^{-rt}\pcheck \,\diff t - e^{-r (\tau_{B^*}^+ \wedge T)} \gamma  1_{\{\tau_{B^*}^+  <\theta \wedge T \}} -  e^{-r \theta} \acheck  1_{\{\tau_{B^*}^+ <\theta < T\}} \right) \right] \\
&= v(x).
\end{align*}
Hence we have (1).  The proof of (2) is similar.  Since $\bar{C}(x, T) \rightarrow C(x)$ as $T \uparrow \infty$, (3) and (4) are also immediate.
\end{proof}

We now analyze the asymptotic behavior of the error functions in terms of $x$ for every fixed $T > 0$.  In view of \eqref{bound_epsilon_v} and \eqref{bound_epsilon_u}, each error bound contains the indicator function $1_{\{\theta \geq T \}}$ and this tends to decrease as $\theta$ decreases (or equivalently, as $x$ decreases).  In particular, if $x$ is sufficiently small and $X$ fluctuates rapidly, these tend to vanish. The following result is immediate by the regularity of a \lev process of unbounded variation; we refer the reader to page 142 of \cite{Kyprianou_2006}.
\begin{proposition}
Fix $T > 0$ and suppose $X$ is of unbounded variation.  Then $\epsilon_v (x, T) \rightarrow 0$ and $\epsilon_u (x, T) \rightarrow 0$ as $x \downarrow 0$.
\end{proposition}

We now consider the limit as $x \uparrow \infty$ as an approximation for the case the maturity $T$ is small in comparison to the default time $\theta$.  In \eqref{bound_epsilon_v} and  \eqref{bound_epsilon_u}, while the error bounds do not converge to zero when $\gamma > 0$, we can obtain the convergence for the case $\gamma = 0$.
By  \eqref{laplace_theta} and Exercise 8.5(i) of \cite{Kyprianou_2006},
\begin{align*}
\lim_{x \rightarrow \infty}\E^x \left[  e^{-r \theta} \acheck  1_{\{\theta \geq T\}} \right] \leq \acheck   \lim_{x \rightarrow \infty}\E^x \left[  e^{-r \theta} \right] = \acheck   \lim_{x \rightarrow \infty} \left( Z^{(r)}(x) - \frac r {\lapinv} W^{(r)}(x) \right) = 0.
\end{align*}  Hence, we can conclude the following limits.
\begin{proposition}
Suppose $\gamma = 0$ and fix $T > 0$.  Then,
\begin{enumerate}
\item $\epsilon_v(x,T) \rightarrow 0$  as $x \uparrow \infty$,
\item $\limsup_{x \uparrow \infty}\epsilon_u(x,T) \leq 0$.
\end{enumerate}
\end{proposition}

\subsection{Limits as $T \downarrow 0$ and Approximation of Stopping Boundaries} \label{subsection_stopping_boundary}

The analytical bounds of the value function obtained above are certainly practical and are indeed tight depending on the chosen parameters.  Next, we consider taking $T \downarrow 0$ and also analyze the stopping boundary as a function of $T$.

We first consider the case $\gamma = 0$, or there is no exercise fee.  Consider the callable case with $\pcheck, \acheck > 0$:
\begin{align*}
\bar{v}(x,\delta) &:=\sup_{\tau \in \S_\delta} \E^x \left[ \int_\tau^{\theta \wedge \delta} e^{-rt}\pcheck \,\diff t -  e^{-r \theta} \acheck  1_{\{\tau <\theta < \delta\}}  \right]
\end{align*}
for any small $\delta > 0$.
Intuitively, for sufficiently small $\delta$, because the movement of the process $X$ until the contract termination can be made arbitrarily small by choosing $\delta$ sufficient small (see \eqref{expectation_delta} below), the investor's strategy is constant; either waiting until contract termination ($\tau = \delta \wedge \theta$) or exercising immediately ($\tau = 0$).   The former gives a zero value. As for the latter case, we observe that
\begin{align}
 \E^x \left[ \int_0^{\theta \wedge \delta} e^{-rt}\pcheck \,\diff t -  e^{-r \theta} \acheck  1_{\{\theta < \delta\}}  \right] = \delta (\pcheck - \acheck \Pi (x, \infty)) + o(\delta)\quad \textrm{as } \delta \downarrow 0 \label{expectation_delta}
\end{align}
because, as in page 247 of Hilberink and Rogers \cite{Hilberink_Rogers_2002},
\begin{align}
 \E^x \left[ e^{-r \theta} 1_{\{\theta < \delta\}}\right] =  \p^x \left\{ \theta < \delta \right\} + o(\delta) = \delta \Pi(x,\infty) + o(\delta) , \quad \textrm{as } \delta \downarrow 0,  \label{rogers_results}
\end{align}
and
\begin{align*}
\E^x \left[ \int_0^{\theta \wedge \delta} e^{-rt} \pcheck \diff t\right] = \E^x \left[ \int_0^{\delta} e^{-rt} \pcheck  \diff t\right]  - \E^x \left[ 1_{\{ \theta \leq \delta \}}\int_{\theta }^\delta e^{-rt} \pcheck  \diff t\right] =  \pcheck  \delta + o(\delta), \quad \textrm{as } \delta \downarrow 0,
\end{align*}
where the last relation holds by \eqref{rogers_results}.  Consequently, $\tau = 0$ if and only if \eqref{expectation_delta} is positive. More precisely, let $C^*$ be such that
\begin{align}
\pcheck - \acheck \Pi(C^*, \infty) = 0, \label{C_star}
\end{align}
if it exists; otherwise, we  set $C^*$ to be zero.
Then the strategy to stop if and only if $x > C^*$ is asymptotically optimal as $\delta$ goes to zero.  By symmetry, for $\bar{u}(x,\delta)$, the strategy to stop if and only if $x < C^*$ is asymptotically optimal.  As a trivial example, suppose $X$ does not have jumps (i.e.\ Brownian motion with drift),  then $C^* = 0$ and exercising immediately is optimal for $\bar{v}$ while waiting until contract termination is optimal for $\bar{u}$ at a time close to maturity.

The asymptotic behavior as $T \downarrow 0$ for the case $\gamma > 0$, on the other hand, is  trivial.  In this case, for both parties, it is never optimal to exercise at a time sufficiently close to the maturity.  This is because as in \eqref{expectation_delta} the premium payment until maturity and the default probability both converge to zero, while the exercise fee is strictly positive.

Using the asymptotic results obtained above, we now analyze the stopping boundary as a function of $T$.  The Markov property of $(t, X_t)$ suggests that the optimal stopping times of $\bar{v}$ and $\bar{u}$ admit the forms
\begin{align*}
\inf \left\{ t \geq 0: X_t \geq B^*(T-t)\right\} \quad \textrm{and} \quad \inf \left\{ t \geq 0: X_t \leq A^*(T-t)\right\},
\end{align*}
for some deterministic functions $A^*(t)$ and $B^*(t)$ that map from $[0,T)$ to $[0,\infty)$.  Thanks to the asymptotic analysis as $T \uparrow \infty$ and $T \downarrow 0$ obtained above, we can actually obtain the asymptotics of $B^*(T)$ and $A^*(T)$ as $T \uparrow \infty$ and $T \downarrow 0$.  As $T \uparrow \infty$, the optimal strategies for the perpetual case, $\tau^+_{B^*} \wedge T$ and  $\tau^-_{A^*} \wedge T$,  are asymptotically optimal for $\bar{v}$ and $\bar{u}$, respectively.  Regarding the asymptotics as $T \downarrow 0$, the asymptotically optimal strategies depend on whether $\gamma = 0$ or $\gamma > 0$.  When $\gamma = 0$, $\tau^+_{C^*} \wedge T$  and $\tau^-_{C^*} \wedge T$   are asymptotically optimal for $\bar{v}$ and $\bar{u}$, respectively.  When $\gamma > 0$, $\theta \wedge T$ is asymptotically optimal for both $\bar{v}$ and $\bar{u}$.  In summary,
\begin{enumerate}
\item $A^*(T) \rightarrow A^*$ and $B^*(T) \rightarrow B^*$ as $T \uparrow \infty$;
\item $A^*(T) \rightarrow C^*$ and $B^*(T) \rightarrow C^*$ as $T \downarrow 0$ when $\gamma = 0$;
\item $A^*(T) \rightarrow 0$ and $B^*(T) \rightarrow \infty$ as $T \downarrow 0$ when $\gamma > 0$.
\end{enumerate}

\begin{figure}[htb]
\begin{center}
\includegraphics[scale=0.6]{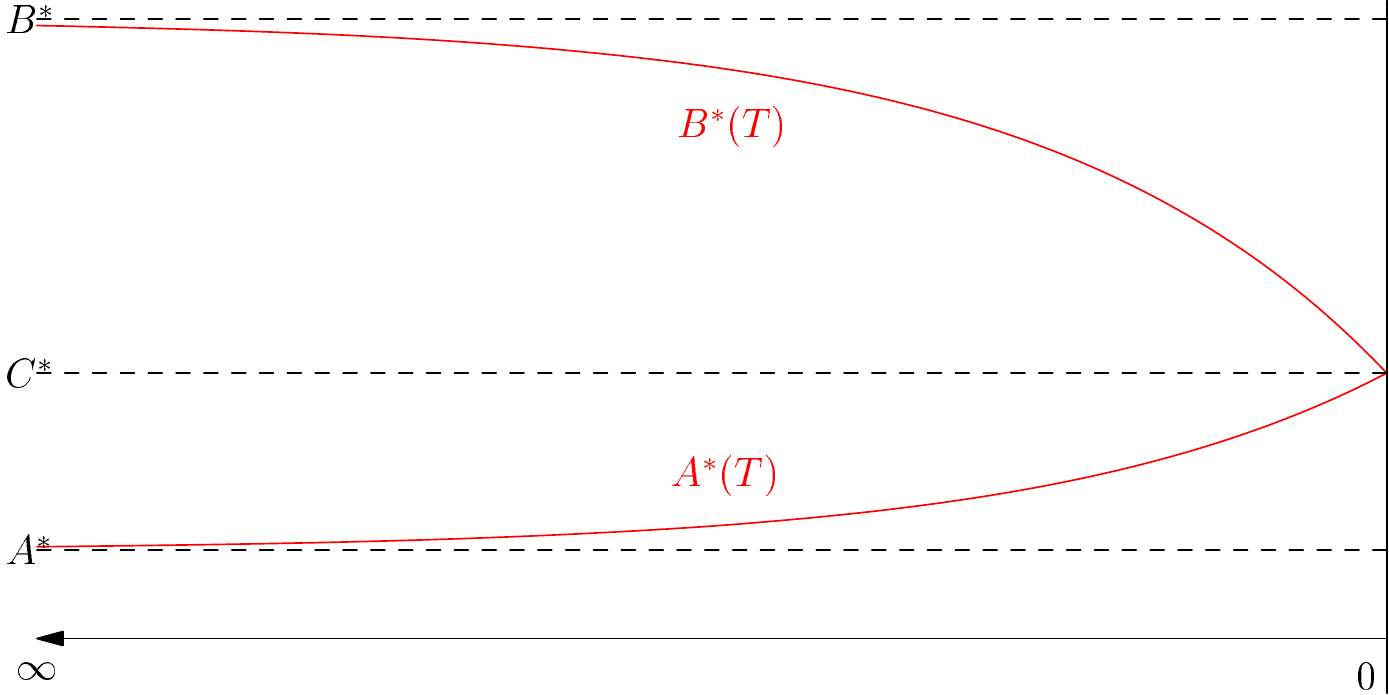} \includegraphics[scale=0.6]{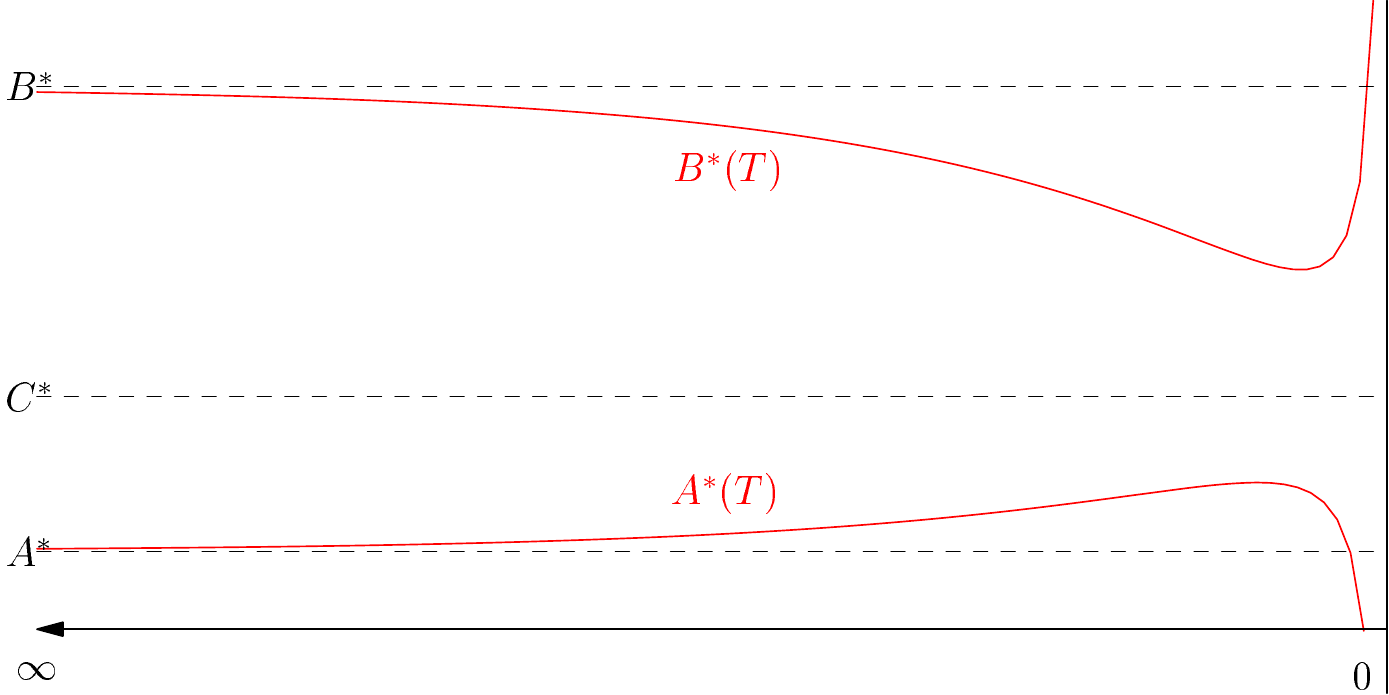} \caption{\small{Illustration of $B^*(T)$ and $A^*(T)$: (left) no exercise fee ($\gamma =0$), and (right) with exercise fee ($\gamma >0$)}.}  \label{figure_finite_horizon}
\end{center}
\end{figure}

In light of these asymptotic behaviors, Figure \ref{figure_finite_horizon} illustrates the shapes of the stopping boundaries for the cases $\gamma =0$ and $\gamma > 0$.  When $\gamma = 0$, $B^*(T)$ tends to increase as $T$ (or the time until maturity) increases while $A^*(T)$ tends to decrease as $T$ increases.  This is commonly observed in finite horizon optimal stopping problems; see Shiryaev and Peskir \cite{Peskir_2006} for the cases of American options, sequential detection and sequential hypothesis testing.  Intuitively speaking, if one has more time till maturity, he does not have to rush to exercise, and hence continuation region tends to increase (decrease) in width as the remaining time increases (decreases).  We can confirm this monotonicity because by the definitions of $A^*$ and $C^*$ as in \eqref{smooth_fit_seller} and \eqref{C_star}
\begin{align*}
C^* > A^*.
\end{align*}
On the other hand, these monotonicities fail once an exercise fee is introduced ($\gamma > 0$).

This observation is particular useful.  Notice that a numerical lower bound can be attained by choosing any feasible strategy and computing its corresponding expected value via simulation.  With this and the analytical upper bounds obtained above, one may obtain a tighter error bound. In particular, for the case $\gamma = 0$, we can focus on the set of monotone functions $B^*(\cdot)$ and $A^*(\cdot)$ connecting $B^*$/$A^*$ and $C^*$.  It is possible to approximate their shapes parametrically or non-parametrically.  For a related technique where non-parametric regression is applied to approximate convex stopping regions,
we refer the reader to, e.g., Section 6 of Dayanik et al.\ \cite{Dayanik_2008}.  Regarding the continuity, monotonicity and smooth/continuous fit for the finite-horizon optimal stopping problem, we refer the reader to \cite{Peskir_2006}.

\subsection{Term Structure of Credit Spreads}  We now consider the credit spread $p^*(x,T)$ that makes $V(x,T) = 0$.
We assume here that $\hat{p}$ is always proportional to $p$ and set as we did in our numerical examples $\hat{p} = q p$ for some constant $q \geq 0$.
 Because the payoff function is monotone in $p$, it is expected that there exists a unique value of credit spread as we have confirmed in our numerical results in Section \ref{section_numerical} for the perpetual case.  The function $p^*(x,T)$ is potentially highly nonlinear both in terms of $x$ and $T$. Nonetheless, we can obtain some asymptotic behaviors as $T \downarrow 0$ and $T \uparrow \infty$ for every fixed $x > 0$ as we describe below.

We first consider the asymptotic behavior of $p^*(x, T)$ as $T \uparrow \infty$ as an approximation to the credit spread for a long maturity.  While the convergence is not guaranteed, the credit spread $p^*(x)$ of the perpetual case can be used as an approximation.  The following holds immediately because both \eqref{asymp_optimality_buyer} and \eqref{asymp_optimality_seller} hold for the case $p = p^*(x)$ and $\hat{p} = q p^*(x)$.

\begin{proposition}For every $x > 0$,
\begin{align*}
\bar{V}(x,T; p^*(x), qp^*(x), \alpha, \ah, \gamma) \rightarrow 0, \\
\bar{U}(x,T; p^*(x), qp^*(x), \alpha, \ah, \gamma) \rightarrow 0,
\end{align*}
as $T \uparrow \infty$.
\end{proposition}

We now consider the asymptotic behavior as $T \downarrow 0$. Suppose $\gamma > 0$, then it is optimal never to exercise close to the maturity.  Therefore, the asymptotic credit spread for all cases (callable/putable and step-up/down) is that of the standard CDS $\lim_{T \downarrow 0}\bar{p}(x, T;\alpha)$.
As in \eqref{def_p_bar}, the credit spread of the standard CDS with finite-maturity is given by
\begin{align*}
\bar{p}(x, T;\alpha)= \frac{\alpha \, r\,\lap_T(x) }{1-\lap_T(x) - e^{-rT} \p^x\{\theta >T\} },\quad \,\text{where}\quad \lap_T(x) := \E^x \left[  e^{-r\theta} 1_{\{\theta \leq T\}}\right].
 \end{align*}
As in page 247 of \cite{Hilberink_Rogers_2002}, we obtain
\begin{align*}
\bar{p}(x, T;\alpha) \rightarrow \alpha \Pi (x,\infty), \quad \textrm{as } ~T \downarrow 0,
\end{align*}
and this is the limit of the credit spread of our default swap as the maturity goes to zero.

The case $\gamma = 0$ is harder to analyze because the boundary $C^*$ depends on $\pcheck = (1-q)p$ which also depends on how the premium $p$ is chosen.  However, because for any $p$ and $\ph$ either stopping immediately or never exercising is asymptotically optimal as $T \downarrow 0$ as we discussed in Section \ref{subsection_stopping_boundary}, the credit spread is expected to converge to either  $\lim_{T \downarrow 0} \bar{p}(x, T; \ah) = \ah \Pi(x, \infty)$ or $\lim_{T \downarrow 0} \bar{p}(x, T; \alpha) = \alpha \Pi(x, \infty)$.

\section{Concluding Remarks} \label{section_conclusion}In summary, the incorporation of American step-up and step-down options give default swap investors the additional flexibility to manage and trade credit risks.  The valuation of these contracts  requires solving for the optimal timing to step-up/down for the protection buyer/seller. The perpetual nature of the contract allows us to compute analytically the investor's optimal exercise threshold under quite general \lev credit risk models. Using the symmetry properties between step-up and step-down contracts, we gain better intuition on various contract specifications, and drastically simplify the procedure to determine the credit spreads. The approximation for the finite-maturity case can be efficiently conducted using our analytical solutions on the perpetual case.

There are a number of avenues for future research. For instance, it would be interesting to value a default swap  where both the protection buyer and seller can terminate the contract early. Then, the valuation problem can be formulated as a modified game option as introduced by Kifer \cite{kifer2000}. In this case, we conjecture that threshold strategies will again be optimal for both parties and constitute Nash or Stackelberg equilibrium \cite{Peskir_2008,Peskir_2009}. Another direction for future research is to consider  derivatives with multiple early exercisable step-up/down options. This is related to some optimal multiple stopping problems arising in other financial applications, such as swing options \cite{Carmona2008} and employee stock options  \cite{LeungSircarESO_MF09}.

\appendix
\section{Proofs}

\subsection{Proof of Lemma \ref{lemma_delta_b}}
Applying the definitions of $\Delta_B(x)$ and $h(x)$ (see (\ref{hx_inproof})) and   noting that $\theta = \infty$  whenever $\tau_B^+ = \infty$, we obtain, for every $x \in (0,B)$,
\begin{align*}
&\Delta_B(x)\\
&=\E^x\left[ 1_{\{\tau_B^+ < \infty\}} \left(\int_{\tau_B^+}^{\theta}
e^{-rt}\pcheck \,\diff t- e^{-r \tau_B^+} \gamma  1_{\{\tau_B^+ <\theta \}}  - e^{-r \theta}\acheck 1_{\{\tau_B^+ <\theta\}} \right) \right] - \E^x \left[ \int_0^\theta e^{-rt} \pcheck \, \diff t\,  - e^{-r \theta} \acheck \, \right] + \gamma \\
&= \E^x\left[ 1_{\{\tau_B^+ < \infty\}} \left(-\int_0^{\tau_B^+}
e^{-rt}\pcheck \,\diff t- e^{-r \tau_B^+} \gamma  1_{\{\tau_B^+ <\theta \}}  + e^{-r \theta}\acheck 1_{\{\tau_B^+ = \theta\}} \right) - 1_{\{\tau_B^+ = \infty\}} \left( \int_0^\theta e^{-rt} \pcheck \, \diff t\,  - e^{-r \theta} \acheck \right) \, \right] + \gamma \\
&= \E^x\left[ 1_{\{\tau_B^+ < \infty\}} e^{-r \tau_B^+}\left( \acheck 1_{\{\tau_B^+ = \theta\}} - \gamma  1_{\{\tau_B^+ <\theta \}}  \right) - \left( \int_0^{\tau_B^+} e^{-rt} \pcheck \, \diff t\, \right) \, \right] + \gamma \\
&= \E^x\left[ 1_{\{\tau_B^+ < \infty\}} e^{-r \tau_B^+}\left( \left(\frac \pcheck r + \acheck \right) 1_{\{\tau_B^+ = \theta\}} + \left( \frac \pcheck r - \gamma  \right) 1_{\{\tau_B^+ <\theta \}}  \right) \right] + \gamma - \frac {\pcheck} r.
\end{align*}

\subsection{Proof of Lemma \ref{lemma_delta_A}}The proof follows from the same arguments for Lemma \ref{lemma_delta_b}, and is thus omitted.

\subsection{Proof of Lemma \ref{lemma_b_star_zero}} When $X$ is of unbounded variation, we have $\varrho(0+) = -\infty$ because $W^{(r)'}(0+) > 0$ and $W^{(r)}(0) = 0$ by Lemma \ref{lemma_zero}.  This implies that it must be of bounded variation for $B^*=0$.  Then, again by Lemma \ref{lemma_zero}, we have after some algebra
\begin{align*}
\varrho(0+) = \frac 1 {\mu} \left( \pcheck - r \gamma -  (\acheck + \gamma)\Pi(0,\infty)\right).
\end{align*}
Since $\varrho(\cdot)$ is increasing, the second condition is equivalent to $\varrho(\cdot)\geq 0$ (i.e. scenario (a)) or $B^*=0$.

\subsection{Proof of Lemma \ref{generator_non_positive}}
As discussed in p.228-229 of \cite{Kyprianou_2006}, for every fixed $0 < x \leq B < \infty$, the stochastic processes
$ \left\{ e^{-r(t \wedge  \tau_{B}^+)} W^{(r)} (X_{t \wedge  \tau_{B}^+}); t \geq 0 \right\}$ and $\left\{e^{-r(t \wedge  \tau_{B}^+)} Z^{(r)} (X_{t \wedge \tau_{B}^+}); t \geq 0 \right\}$ are $\p^x$-martingales.  Consequently, we have
\begin{align}
(\mathcal{L}-r) W^{(r)}(x) = (\mathcal{L}-r) Z^{(r)}(x) = 0, \quad x > 0. \label{generator_zero}
\end{align}
 With this and the definition of $v_{B^*}$ in \eqref{U-smaller_than_b_star},  if $0 < B^* < \infty$, it follows that
\begin{align}
(\mathcal{L} - r) v_{B^*} (x) = 0 \quad x\in(0, B^*). \label{eqn_generator_equal}
\end{align}

It remains to show that $(\mathcal{L}-r) v_{B^*}(x) \leq 0$ for every $x > B^*$. We do so via Lemmas \ref{generator_less_b_star_zero}-\ref{lemma_delta_decreasing} below.
\begin{lemma} \label{generator_less_b_star_zero}
If $B^* = 0$, then we have $(\mathcal{L}-r) v_{B^*}(0+) = (\mathcal{L}-r) h(0+) \leq 0$.
\end{lemma}
\begin{proof}
Using \eqref{generator_zero}, we have $(\mathcal{L}-r) h(0+) = (\mathcal{L}-r) L(0+)$, where
$L(x) := ( \frac \pcheck r - \gamma ) + \left(\acheck + \gamma \right) 1_{\{x \leq 0\}}$, $x \in \R$. Now, with $L'(0+)=0$, plus $\sigma = 0$ and $\Pi(0,\infty) < \infty$ by Lemma \ref{lemma_b_star_zero}, we obtain
\begin{align*}
 (\mathcal{L}-r) L(0+) &=  \int_0^\infty \left( L(0-z) - L(0) \right) \Pi (\diff z) - r \left( \frac \pcheck r - \gamma \right) =  (\acheck+\gamma) \Pi(0,\infty) - \left( \pcheck - \gamma r \right).
\end{align*}
Recalling Lemma \ref{lemma_b_star_zero}, we conclude that  $(\mathcal{L}-r) h(0+) \le 0$.
\end{proof}

\begin{lemma} \label{discont_generator}
If $\sigma > 0$ and $0 < B^* < \infty$, then we have $\Delta''(B^*-) \geq 0$.
\end{lemma}
\begin{proof}
By (\ref{definition_b_star}) and direct computation, we get
\begin{align*}
\Delta''(B^*-)= \left( r \acheck + \pcheck \right) W^{(r)'} (B^*)  - \frac {W^{(r)''}(B^*)} {W^{(r)}(B^*)} G^{(r)}(B^*)
=  - \left( r \acheck + \pcheck \right) \frac {(W^{(r)} (B^*))^2} {W^{(r)'} (B^*)} \frac \partial {\partial B}\frac {W^{(r)'}(B)} {W^{(r)}(B)},
\end{align*}
which is non-negative by \eqref{assumeW}. Here $W^{(r)}$ is twice differentiable since $\sigma > 0$.
\end{proof}

Suppose $0 < B^* < \infty$. Since $v_{B^*}(\cdot)$ and $v_{B^*}'(\cdot)$ are continuous at $B^*$ by the continuous and smooth fit conditions, it follows from Lemma \ref{discont_generator}   that  $(\mathcal{L}-r) v_{B^*}(B^*-) \geq (\mathcal{L}-r) v_{B^*}(B^*+)$, and then by \eqref{eqn_generator_equal} that $(\mathcal{L}-r) v_{B^*}(B^*+) \leq 0$.
Suppose $B^* = 0$, we have $(\mathcal{L}-r) v_{B^*}(B^*+) \leq 0$ by Lemma \ref{generator_less_b_star_zero}. Now, in order to show $(\mathcal{L}-r) v_{B^*}(x) \leq 0$ on $(B^*,\infty)$, it is sufficient to prove that $(\mathcal{L}-r) v_{B^*}(x)$ is decreasing on $(B^*,\infty)$.
To this end, we rewrite $v_{B^*}$ as
\begin{align*}
v_{B^*}(x) = \left( \widetilde{h}(x) + \widetilde{\Delta}(x) \right) 1_{\{x \neq 0\}}, \quad x \in \mathbb{R}
\end{align*}
where
\begin{align}
\notag \widetilde{h}(x) &= - \left( \frac \pcheck r + \acheck \right) \left( Z^{(r)}(x) - \frac r {\lapinv} W^{(r)}(x) \right), \quad x \in \mathbb{R},\\
\label{tildeDel}\widetilde{\Delta}(x) &= \left\{ \begin{array}{ll}\displaystyle \left( \frac \pcheck r + \acheck \right) Z^{(r)}(x) - \frac {W^{(r)}(x)} {W^{(r)}(B^*)} G^{(r)}(B^*), &  x \in (-\infty, B^*),\\
\displaystyle \frac \pcheck r - \gamma, & x \in [B^*, \infty).
\end{array} \right.
\end{align}
By \eqref{generator_zero}, $(\mathcal{L} - r) \widetilde{h}(x)=0$ for every $x > 0$. Furthermore, $\widetilde{\Delta}'(x)=\widetilde{\Delta}''(x)=0$ on $x > B^*$,
and hence
\begin{align*}
(\mathcal{L} - r) v_{B^*}(x) = (\mathcal{L} - r) \widetilde{\Delta} (x) = \int_0^\infty \left( \widetilde{\Delta} (x-z) - \widetilde{\Delta} (x) \right) \Pi(\diff z) - \left( \pcheck - r \gamma \right), \quad x > B^*.
\end{align*}
 In order to show that this is decreasing in $x$ on $(B^*,\infty)$, it is sufficient to show that the integrand in the right-hand side is decreasing in $x$ or equivalently $\widetilde{\Delta}(x-z)$ is decreasing in $x$ for every fixed $z$ by noting that  $\widetilde{\Delta}(x)$ is a constant on $(B^*,\infty)$.

\begin{lemma} \label{lemma_delta_decreasing}
 The function $\widetilde{\Delta}(\cdot)$ is decreasing on $\mathbb{R}$ and is uniformly bounded  below  by $ \pcheck / r - \gamma > 0$.
\end{lemma}
\begin{proof}
It is clear that $\widetilde{\Delta}(\cdot)$ in \eqref{tildeDel} is
monotonically decreasing when $B^*=0$ because $\widetilde{\Delta}(0-)
= \frac \pcheck r + \acheck > \frac \pcheck r - \gamma =
\widetilde{\Delta}(0+)$. Suppose $0 < B^* < \infty$. By differentiating
$\widetilde{\Delta}(\cdot)$, for $ 0 < x < B^*$, we get
\begin{align*}
\widetilde{\Delta}'(x) &= \left( \pcheck + \acheck r \right) W^{(r)} (x) - \frac {W^{(r)'}(x)} {W^{(r)}(B^*)} G^{(r)}(B^*) = W^{(r)'}(x) \left\{ \left( r \acheck + \pcheck \right) \left[ \frac {W^{(r)} (x)} {W^{(r)'}(x)} - \frac {W^{(r)} (B^*)} {W^{(r)'} (B^*)} \right] \right\}.
\end{align*}

Note that $\widetilde{\Delta}'(\cdot)$ is non-positive because \eqref{assumeW} implies that
\begin{align*}
\frac {W^{(r)'}(x)} {W^{(r)}(x)} \geq \frac {W^{(r)'}(B^*)} {W^{(r)}(B^*)} \Longleftrightarrow \frac {W^{(r)}(x)} {W^{(r)'}(x)} \leq \frac {W^{(r)}(B^*)} {W^{(r)'}(B^*)}, \quad x \leq B^*.
\end{align*}
Furthermore, at zero, we have $\widetilde{\Delta} (0+) - \widetilde{\Delta} (0-) = -\frac {W^{(r)}(0)} {W^{(r)}(B^*)} G^{(r)}(B^*) \leq 0.$
As for uniform (lower) boundedness,  continuous fit implies that $\widetilde{\Delta}(B^*-)=\widetilde{\Delta}(B^*+)=\frac \pcheck r -\gamma > 0$ (see \eqref{assumption_basic}). Since $\widetilde{\Delta}$ is constant on $(-\infty,0) \cup (B^*,\infty)$, we conclude.
\end{proof}

We therefore have $(\mathcal{L}-r) v_{B^*} (x) \leq 0$ for every $x > B^*$, with $0 \leq B^* < \infty$. Along with \eqref{eqn_generator_equal}, the proof of Lemma \ref{generator_non_positive}  is complete.

\subsection{Proof of Theorem \ref{optimality_buyer}}
Due to the potential discontinuity and non-smoothness of the value function at zero, we need to proceed carefully.
Before showing the main result, we first prove that   $0 = v_{B^*}(0) \leq v_{B^*}(0+)$ as illustrated in Figure \ref{plot_v_h}.

Suppose $0 < B^* < \infty$. Both $h(x)$ and $v_{B^*}(x)$ are both increasing in $x$ (see \eqref{hx} and \eqref{U-smaller_than_b_star}). Since $W^{(r)}$ is increasing and non-negative, we must have $\frac {\pcheck + \acheck r} {\Phi(r)}  - \frac {G^{(r)}(B^*)} {W^{(r)}(B^*)} \geq 0$ and hence $v_{B^*}(0+) = W^{(r)}(0) \left( \frac {\pcheck + \acheck r} {\Phi(r)}  - \frac {G^{(r)}(B^*)} {W^{(r)}(B^*)} \right) \geq 0$,
and this is strictly positive if and only if $X$ has paths of bounded variation by Lemma \ref{lemma_zero}. On the other hand, if  $B^*=0$ (which implies $X$ is of bounded variation by Lemma \ref{lemma_b_star_zero}), we also have $v_{B^*}(0+) \geq 0$ because \eqref{vb_derv} implies that, for any $\varepsilon > 0$, $\partial v_B(\varepsilon) /  {\partial B} \leq 0$ for every $B > \varepsilon$ and $\lim_{B \rightarrow \infty}v_B(\varepsilon) = 0$.

We focus on the case it is discontinuous at $0$ (or $X$ is of bounded variation) and then address how the proof can be modified for the other case.
We first construct a sequence of functions $v_n(\cdot)$ such that (1) it is continuous everywhere, (2) $v_n(x) = v_{B^*}(x)$ on $x \in (0,\infty)$ and (3) $v_n(x) \downarrow v_{B^*}(x)$ pointwise for every fixed $x \in (-\infty,0)$.
 This implies, by noting that $v_{B^*}(x)=v_n(x)$ and $v_{B^*}'(x)=v'_n(x)$ on $(0,\infty)$, that $(\mathcal{L}-r) (v_n - v_{B^*})(x)$ decreases monotonically in $n$ to zero  for every fixed $x \in (0,\infty) $ by the monotone convergence theorem.  Notice that $v_{B^*}(\cdot)$ is uniformly bounded  because $h(\cdot)$ is.  Hence, we can choose so that $v_n$ is also uniformly bounded for every fixed $n \geq 1$.

Suppose $K < \infty$ is the maximum difference between $v_{B^*}$ and $v_n$.  Using $N$ as the Poisson random measure for $-X$, we have by compensation formula (see, e.g., Theorem 4.4 in \cite{Kyprianou_2006}), for every $\nu \in \S$,

\begin{multline}
\E^x \left[ \int_0^{\nu} e^{-rs} |(\mathcal{L}-r) (v_n-v_{B^*}) (X_{s-})| \diff s\right]  \leq K \E^x \left[ \int_0^{\theta} e^{-rs} \Pi(X_{s-},\infty) \diff s\right]  \\ = K \E^x \left[ \int_0^\infty \int_0^\infty e^{-rs} 1_{\{\theta \geq s, \; u > X_{s-}\}} \Pi(\diff u) \diff s \right]  = K \E^x \left[ \int_0^\infty \int_0^\infty e^{-rs} 1_{\{\theta \geq s, \; u > X_{s-}\}} N (\diff u \times \diff s) \right] \\ = K \E^x \left[ e^{-r \theta} 1_{\{X_\theta < 0, \, \theta < \infty \}}\right] < \infty. \label{upper_bound_generator_v_n}
\end{multline}

 Fix $\epsilon > 0$.  By applying Ito's formula to $\left\{ e^{-r {(t \wedge  \tau_\epsilon^-)}} v_n(X_{t \wedge \tau_\epsilon^-}); t \geq 0 \right\}$, we see that
\begin{align}
\left\{ e^{-r {(t  \wedge \tau_\epsilon^-)}} v_n(X_{t \wedge \tau_\epsilon^-}) - \int_0^{t  \wedge \tau_\epsilon^-} e^{-rs}\left( (\mathcal{L} - r) v_n (X_{s-})  \right) \diff s; \quad t \geq 0 \right\} \label{local_martingale}
\end{align}
is a local martingale.  Suppose $\left\{\sigma_k; k \geq 1 \right\}$ is the corresponding localizing sequence, namely, for given $k \geq 1$,
\begin{multline*}
\E^x \left[ e^{-r {(t  \wedge \tau_\epsilon^- \wedge \sigma_k)}} v_n(X_{t  \wedge \tau_\epsilon^-\wedge \sigma_k}) \right] = v_n(x) + \E^x \left[  \int_0^{t  \wedge \tau_\epsilon^-\wedge \sigma_k} e^{-rs}\left( (\mathcal{L} - r) v_n (X_{s-})  \right) \diff s \right]  \\
= v_n(x) + \E^x \left[  \int_0^{t  \wedge \tau_\epsilon^-\wedge \sigma_k} e^{-rs}\left( (\mathcal{L} - r) (v_n - v_{B^*}) (X_{s-})  \right) \diff s \right] + \E^x \left[  \int_0^{t  \wedge \tau_\epsilon^-\wedge \sigma_k} e^{-rs}\left( (\mathcal{L} - r) v_{B^*} (X_{s-})  \right) \diff s \right],
\end{multline*}
where the second line makes sense by \eqref{upper_bound_generator_v_n}.  Applying the dominated convergence theorem on the left-hand side and the monotone convergence convergence theorem on the right-hand side (here the integrands in the two expectations are positive and negative, respectively, by Lemma \ref{generator_non_positive}), we obtain again by \eqref{upper_bound_generator_v_n}
\begin{align*}
\E^x \left[ e^{-r {(t  \wedge \tau_\epsilon^-)}} v_n(X_{t \wedge
\tau_\epsilon^-}) \right] = v_n(x) + \E^x \left[  \int_0^{t  \wedge \tau_\epsilon^-} e^{-rs}\left( (\mathcal{L} - r) v_n (X_{s-})  \right) \diff s \right].
\end{align*}
Hence,  \eqref{local_martingale} is in fact a martingale.

Now fix $\nu \in \S$. By the optional sampling theorem, we have, for any $M \geq 0$, that
\begin{multline*}
\E^x \left[ e^{-r {(\nu \wedge \tau_\epsilon^- \wedge M)}} v_n(X_{\nu \wedge \tau_\epsilon^- \wedge M}) \right] = v_n(x) + \E^x \left[  \int_0^{\nu \wedge \tau_\epsilon^- \wedge M} e^{-rs}\left( (\mathcal{L} - r) v_n (X_{s-})  \right) \diff s \right] \\
\leq v_n(x) + \E^x \left[  \int_0^{\nu \wedge \tau_\epsilon^- \wedge M} e^{-rs}\left( (\mathcal{L} - r) (v_n-v_{B^*}) (X_{s-})  \right) \diff s \right]
\end{multline*}
where the last equality holds by Lemma \ref{generator_non_positive}.  Applying dominated convergence via \eqref{upper_bound_generator_v_n} and because $\nu \wedge \tau_\epsilon^- \xrightarrow{a.s.} \nu$ as $\epsilon \downarrow 0$, upon taking limits as $M \uparrow \infty$ and  $\epsilon \downarrow 0$,
\begin{align}
\E^x \left[ e^{-r \nu} v_n(X_{\nu}) 1_{\{\nu < \infty \}}\right] \leq v_n(x) + \E^x \left[ \int_0^{\nu} e^{-rs} ((\mathcal{L}-r) (v_n -v_{B^*})(X_{s-})) \diff s\right]. \label{supermtg_proof}
\end{align}
Furthermore, the monotone convergence theorem yields  the followings:
\begin{align*}
&\lim_{n \rightarrow \infty} \E^x \left[ e^{-r \nu} v_n(X_\nu) 1_{\{\nu < \infty \}} \right] =  \E^x \left[ e^{-r \nu} \left( v_{B^*}(X_{\nu }) + v_{B^*}(0+) 1_{\{X_\nu = 0, \nu < \infty\}} \right) \right] \geq \E^x \left[ e^{-r \nu} v_{B^*}(X_{\nu }) 1_{\{\nu < \infty \}} \right],\\
&\lim_{n \rightarrow \infty} \E^x \left[ \int_0^{\nu} e^{-rs} ((\mathcal{L}-r) (v_n -v_{B^*})(X_{s-})) \diff s\right] = 0.
\end{align*}
Therefore, by taking $n \rightarrow \infty$ on both sides of \eqref{supermtg_proof} (note $v_{B^*}(x) = v_n(x)$), we have
\begin{align*}
v_{B^*}(x)\geq \E^x \left[ e^{-r \nu} v_{B^*}(X_{\nu}) 1_{\{\nu < \infty \}} \right] \geq  \E^x \left[ e^{-r \nu}h(X_{\nu}) 1_{\{\nu < \infty \}} \right], \quad \nu \in \S,
\end{align*}
where the last inequality follows from  \eqref{v_greater_than_h}. This together with the fact that the stopping time $\tau_{B^*}^+$ corresponds to the value function $v_{B^*}$ completes the proof for the case $v_{B^*}$ is discontinuous at $0$.  For the case $v_{B^*}$ is continuous (unbounded variation case), the proof is simpler; the approximation via $v_n$ is no longer needed but the localization via $\tau_\epsilon^-$ is still needed because the value function fails to be $C^1$ at zero.

\subsection{Proof of Lemma \ref{u_minimum_seller}}
Fix $A > 0$. Using   \eqref{laplace_theta}  and Lemma \ref{remark_gamma_x_a}, we differentiate to get
\begin{align*}
\frac \partial {\partial A} \zeta(x-A) &= - r W^{(r)}(x-A) + \frac r {\lapinv} W^{(r)'}(x-A), \\
\frac \partial {\partial A} \Gamma(x;A)
&= - \frac {\rho(A)} {\Phi(r)} W^{(r)'} (x-A) +  W^{(r)} (x-A) \int_A^\infty \Pi(\diff u) \left( 1 - e^{-\Phi(r) (u-A)}\right) \\
&=  \left( W^{(r)}(x-A) - \frac 1 {\lapinv} W^{(r)'}(x-A)\right) \rho(A).
\end{align*}
Therefore, it follows from (\ref{delta_seller}) that
\begin{align}
\frac \partial {\partial A} \Delta_A(x) &= \left( W^{(r)}(x-A) - \frac 1 {\lapinv} W^{(r)'}(x-A)\right) \left((\gamma r + \pcheck) - (\acheck-\gamma) \rho(A) \right)\notag\\
&= \frac{e^{\lapinv(x-A)}} {\lapinv} W'_{\lapinv}(x-A) \left( (\acheck-\gamma) \rho(A) - (\gamma r + \pcheck)\right),\label{App_sign}
\end{align}
where $W_{\lapinv} (x):=e^{-\lapinv x} W^{(r)} (x) $. Since $W_{\lapinv}(x)$ is increasing in $x$ (see  p.219 of \cite{Kyprianou_2006}) and   $\rho(A)$ is decreasing, we can infer the sign of $\frac \partial {\partial A} \Delta_A(x)$ from \eqref{signofderiv}.

\subsection{Proof of Lemma \ref{lemma_generator_zero_seller}}
(1) We first show  $(\mathcal{L}-r) u_{A^*}(x) = 0$ for every $x > A^*$.  Define a function $\widetilde{u}$  by \eqref{u_tilde_a} and (\ref{def_w2}) but extended to  the whole real line. This way, for $x < 0$, $u_{A^*}(x) = 0$  while $\widetilde{u}(x)$  does not vanish. Now fix $A^* > 0$. Using (\ref{u_tilde_a}) and recalling that $(\mathcal{L}-r) Z^{(r)}(y)=(\mathcal{L}-r) W^{(r)}(y) = 0$ for every $y > 0$, we have
\begin{align}
(\mathcal{L}-r)\widetilde{u}(x) &= (\acheck - \gamma) \left[  \frac 1 r \int_{A^*}^\infty \Pi (\diff u)   \left[ (\mathcal{L}-r)Z^{(r)} (x-A^*) - (\mathcal{L}-r)Z^{(r)}(x-u) \right]  \right]\notag\\
&- \left( \gamma + \frac \pcheck r \right)  (\mathcal{L}-r)Z^{(r)}(x-A^*) + \left(\frac \pcheck r + \acheck \right) \left[ (\mathcal{L}-r)Z^{(r)}(x) - \frac r {\lapinv}(\mathcal{L}-r) W^{(r)}(x) \right] \notag\\
&= - \frac {\acheck - \gamma} r \int_x^\infty \Pi (\diff u)   \left[ (\mathcal{L}-r)Z^{(r)}(x-u) \right]   = (\acheck - \gamma) \Pi(x,\infty). \label{Lruu1}
\end{align}
Here the generator can go into the integral because, by Lemma \ref{lemma_zero}, $Z^{(r)}(\cdot)$ is continuous when $\sigma = 0$ and is continuously differentiable when $\sigma > 0$.
On the other hand, if $A^*=0$, then from (\ref{def_w2}) we derive that
$(\mathcal{L}-r)\widetilde{u}(x)  = (\acheck - \gamma) \Pi(x,\infty)$
as above. Furthermore, it is straightforward to show that $(\mathcal{L}-r)(u_{A^*}(x)-\widetilde{u}(x)) = -\Pi(x,\infty) (\acheck-\gamma)$.  Combining this with \eqref{Lruu1}, we get $(\mathcal{L}-r) u_{A^*}(x) = 0$.

(2) We now show for the case $A^* > 0$ that we have $(\mathcal{L}-r) u_{A^*}(x) > 0$ for every $0 < x < A^*$. Notice that $u_{A^*}(x) = g(x)$ on $0 < x < A^*$.  By \eqref{def_G},
\begin{align*}
(\mathcal{L}-r) u_{A^*}(x) = (\mathcal{L}-r) L(x) \quad \textrm{where} \quad
L(x) := - \left(\gamma + \frac \pcheck r \right) - (\acheck - \gamma) 1_{\{x \leq 0\}}.
\end{align*}
Therefore, we have
\begin{multline*}
(\mathcal{L}-r) u_{A^*}(x) = (\mathcal{L}-r) L(x) = \int_0^\infty \left( L(x-u) - L(x) \right) \Pi(\diff u)  - r L(x) \\ =  - \left( \acheck - \gamma \right) \Pi(x,\infty)  + (r \gamma + \pcheck) \leq - \left( \acheck - \gamma \right) \rho(A^*) + (r \gamma + \pcheck)  = 0.
\end{multline*}
Here the inequality holds by (\ref{varrho_A}) and because $x < A^*$. The last equality holds by \eqref{smooth_fit_seller}.

\linespread{1.1}
\bibliographystyle{abbrv}
\bibliographystyle{agsm}
\begin{small}
 \bibliography{Leung_Yamazaki_bib}

\def\cprime{$'$}
\begin{thebibliography}{10}

\bibitem{Alili2005}
L.~Alili and A.~E. Kyprianou.
\newblock Some remarks on first passage of {L}\'evy processes, the {A}merican
  put and pasting principles.
\newblock {\em Ann. Appl. Probab.}, 15(3):2062--2080, 2005.

\bibitem{Asmussen_2004}
S.~Asmussen, F.~Avram, and M.~R. Pistorius.
\newblock Russian and {A}merican put options under exponential phase-type
  {L}\'evy models.
\newblock {\em Stochastic Process. Appl.}, 109(1):79--111, 2004.

\bibitem{AsmussenMadanPistorius07}
S.~Asmussen, D.~Madan, and M.~Pistorius.
\newblock Pricing equity default swaps under an approximation to the {CGMY}
  {L}evy model.
\newblock {\em J. Comput.\ Finance}, 11(2):79--93, 2007.

\bibitem{avramchanusabel02}
F.~Avram, T.~Chan, and M.~Usabel.
\newblock On the valuation of constant barrier options under spectrally
  one-sided exponential {L}\'{e}vy models and {C}arr's approximation for
  {A}merican puts.
\newblock {\em Stochastic Process. Appl.}, 100:75--107, 2002.

\bibitem{Avram_2004}
F.~Avram, A.~E. Kyprianou, and M.~R. Pistorius.
\newblock Exit problems for spectrally negative {L}\'evy processes and
  applications to ({C}anadized) {R}ussian options.
\newblock {\em Ann. Appl. Probab.}, 14(1):215--238, 2004.

\bibitem{Avram_et_al_2007}
F.~Avram, Z.~Palmowski, and M.~R. Pistorius.
\newblock On the optimal dividend problem for a spectrally negative {L}\'evy
  process.
\newblock {\em Ann. Appl. Probab.}, 17(1):156--180, 2007.

\bibitem{Barndorff_1998}
O.~E. Barndorff-Nielsen.
\newblock Processes of normal inverse {G}aussian type.
\newblock {\em Finance Stoch.}, 2(1):41--68, 1998.

\bibitem{Baurdoux2008}
E.~Baurdoux and A.~E. Kyprianou.
\newblock The {M}c{K}ean stochastic game driven by a spectrally negative
  {L}\'evy process.
\newblock {\em Electron. J. Probab.}, 13:no. 8, 173--197, 2008.

\bibitem{Baurdoux2009}
E.~Baurdoux and A.~E. Kyprianou.
\newblock The {S}hepp-{S}hiryaev stochastic game driven by a spectrally
  negative l\'evy process.
\newblock {\em Theory Probab.\ Appl.}, 53, 2009.

\bibitem{BlackCox76}
F.~Black and J.~Cox.
\newblock Valuing corporate securities: Some effects of bond indenture
  provisions.
\newblock {\em J.\ Finance}, 31:351--367, 1976.

\bibitem{Boyarchenko_QF04}
S.~Boyarchenko.
\newblock Pricing of perpetual {B}ermudan options.
\newblock {\em Quant.\ Finance}, 4:525--547, 2004.

\bibitem{schoutensCDS07}
J.~Cariboni and W.~Schoutens.
\newblock Pricing credit default swaps under {L}\'evy models.
\newblock {\em J. Comput.\ Finance}, 10(4):1--21, 2007.

\bibitem{Carmona2008}
R.~Carmona and N.~Touzi.
\newblock Optimal multiple stopping and valuation of swing options.
\newblock {\em Math. Finance}, 18(2):239--268, April 2008.

\bibitem{CGMY}
P.~Carr, H.~Geman, D.~B. Madan, and M.~Yor.
\newblock The fine structure of asset returns: An empirical investigation.
\newblock {\em J.\ Bus.}, 75(2):305--332, 2002.

\bibitem{Chan_2009}
T.~Chan, A.~Kyprianou, and M.~Savov.
\newblock Smoothness of scale functions for spectrally negative {L}\'evy
  processes.
\newblock {\em Probab.\ Theory Related Fields}, 150:691--708, 2011.

\bibitem{Dayanik_2008}
S.~Dayanik, C.~Goulding, and H.~V. Poor.
\newblock Bayesian sequential change diagnosis.
\newblock {\em Math. Oper. Res.}, 33(2):475--496, 2008.

\bibitem{EberleinKlugeSchb06}
E.~Eberlein, W.~Kluge, and P.~J. Sch\"{o}nbucher.
\newblock The {L}\'{e}vy {L}ibor model with default risk.
\newblock {\em J.\ Credit Risk}, 2(2):3--42, 2006.

\bibitem{Egami_Yamazaki_2010_2}
M.~Egami and K.~Yamazaki.
\newblock Phase-type fitting of scale functions for spectrally negative
  {L}\'evy processes.
\newblock {\em Working paper. ArXiv:1005.0064}, 2012.

\bibitem{Egami_Yamazaki_2010}
M.~Egami and K.~Yamazaki.
\newblock Precautionary measures for credit risk management in jump models.
\newblock {\em Stochastics}, To appear.

\bibitem{Peskir_2008}
E.~Ekstr{\"o}m and G.~Peskir.
\newblock Optimal stopping games for {M}arkov processes.
\newblock {\em SIAM J.\ Control Optim.}, 47(2):684--702, 2008.

\bibitem{Feldmann_1998}
A.~Feldmann and W.~Whitt.
\newblock Fitting mixtures of exponentials to long-tail distributions to
  analyze network performance models.
\newblock {\em Performance Evaluation}, (31):245--279, 1998.

\bibitem{helberenkrogers}
B.~Hilberink and C.~Rogers.
\newblock Optimal capital structure and endogenous default.
\newblock {\em Finance Stoch.}, 6(2):237--263, 2002.

\bibitem{Hilberink_Rogers_2002}
B.~Hilberink and L.~C.~G. Rogers.
\newblock Optimal capital structure and endogenous default.
\newblock {\em Finance Stoch.}, 6(2):237--263, 2002.

\bibitem{hirsaMadan}
A.~Hirsa and D.~Madan.
\newblock Pricing american options under variance gamma.
\newblock {\em J. Comput.\ Finance}, 7(2), 2003.

\bibitem{HullWhiteCDswaption03}
J.~Hull and A.~White.
\newblock The valuation of credit default swap options.
\newblock {\em J.\ Derivatives}, 10(3):40--50, 2003.

\bibitem{jaimungalSukov08}
K.~Jackson, S.~Jaimungal, and V.~Surkov.
\newblock {F}ourier space time stepping for option pricing with {L}\'{e}vy
  models.
\newblock {\em J. Comput.\ Finance}, 12(2):1--29, 2008.

\bibitem{schoutenjosson}
H.~J\"{o}sson and W.~Schoutens.
\newblock Single name credit default swaptions meet single sided jump models.
\newblock {\em Rev.\ Derivatives Res.}, 11(1):153--169, 2008.

\bibitem{kifer2000}
Y.~Kifer.
\newblock Game options.
\newblock {\em Finance Stoch.}, 4:443--463, 2000.

\bibitem{Kou_Wang_2004}
S.~G. Kou and H.~Wang.
\newblock Option pricing under a double exponential jump diffusion model.
\newblock {\em Manage. Sci.}, 50(9):1178--1192, 2004.

\bibitem{Kyprianou_2006}
A.~E. Kyprianou.
\newblock {\em Introductory lectures on fluctuations of {L}\'evy processes with
  applications}.
\newblock Universitext. Springer-Verlag, Berlin, 2006.

\bibitem{Kyprianou_Palmowski_2007}
A.~E. Kyprianou and Z.~Palmowski.
\newblock Distributional study of de {F}inetti's dividend problem for a general
  {L}\'evy insurance risk process.
\newblock {\em J. Appl. Probab.}, 44(2):428--448, 2007.

\bibitem{Kyprianou_Surya_2007}
A.~E. Kyprianou and B.~A. Surya.
\newblock Principles of smooth and continuous fit in the determination of
  endogenous bankruptcy levels.
\newblock {\em Finance Stoch.}, 11(1):131--152, 2007.

\bibitem{LeungSircarESO_MF09}
T.~Leung and R.~Sircar.
\newblock Accounting for risk aversion, vesting, job termination risk and
  multiple exercises in valuation of employee stock options.
\newblock {\em Math.\ Finance}, 19(1):99--128, 2009.

\bibitem{levendorski_QF04}
S.~Z. Levendorski.
\newblock Early exercise boundary and option prices in l\'{e}vy driven models.
\newblock {\em Quant.\ Finance}, 4(5):1469--7696, 2004.

\bibitem{Loeffen_2008}
R.~L. Loeffen.
\newblock On optimality of the barrier strategy in de {F}inetti's dividend
  problem for spectrally negative {L}\'evy processes.
\newblock {\em Ann. Appl. Probab.}, 18(5):1669--1680, 2008.

\bibitem{Madan_1998}
D.~Madan, C.~P.P., and C.~E.C.
\newblock The variance gamma processes and option pricing.
\newblock {\em Europ.\ Finance Rev.}, 2:79--105, 1998.

\bibitem{Merton_1976}
R.~Merton.
\newblock Option pricing when underlying stock returns are discontinuous.
\newblock {\em J. Financial Econom}, 3:125--144, 1976.

\bibitem{mordecki_FS02}
E.~Mordecki.
\newblock Optimal stopping and perpetual options for l\'{e}vy processes.
\newblock {\em Finance Stoch.}, 4:473--493, 2002.

\bibitem{Peskir_2001}
G.~Peskir.
\newblock {\em Principles of optimal stopping and free-boundary problems},
  volume~68 of {\em Lecture Notes Series (Aarhus)}.
\newblock University of Aarhus, Department of Mathematics, Aarhus, 2001.
\newblock Dissertation, University of Aarhus, Aarhus, 2002.

\bibitem{Peskir_2009}
G.~Peskir.
\newblock Optimal stopping games and {N}ash equilibrium.
\newblock {\em Theory Probab. Appl.}, 53(3):558--571, 2009.

\bibitem{Peskir_2006}
G.~Peskir and A.~Shiryaev.
\newblock {\em Optimal stopping and free-boundary problems}.
\newblock Lectures in Mathematics ETH Z\"urich. Birkh\"auser Verlag, Basel,
  2006.

\bibitem{Peskir_2000}
G.~Peskir and A.~N. Shiryaev.
\newblock Sequential testing problems for {P}oisson processes.
\newblock {\em Ann. Statist.}, 28(3):837--859, 2000.

\bibitem{Peskir_2002}
G.~Peskir and A.~N. Shiryaev.
\newblock Solving the {P}oisson disorder problem.
\newblock In {\em Advances in finance and stochastics}, pages 295--312.
  Springer, Berlin, 2002.

\bibitem{zhou2001}
C.~Zhou.
\newblock The term structure of credit spreads with jump risk.
\newblock {\em J.\ Banking Finance}, 25:2015--2040, 2001.

\end{thebibliography}
\end{small}
 \end{document}